\documentclass{article}%
\usepackage{amsmath}
\usepackage{amsfonts}
\usepackage{amssymb}
\usepackage{graphicx}%
\providecommand{\U}[1]{\protect\rule{.1in}{.1in}}
\newtheorem{theorem}{Theorem}

\newtheorem{lemma}[theorem]{Lemma}

\newenvironment{proof}[1][Proof]{\noindent\textbf{#1.} }{\ \rule{0.5em}{0.5em}}
\begin{document}

\title{An asymptotic theory of cloning of classical state families}
\author{Keiji Matsumoto \ \ \ keiji@nii.ac.jp\\National Institute of Infroamatics}
\maketitle

\begin{abstract}
Cloning, or approximate cloning, is one of basic operations in quantum
information processing. \ In this paper, we deal with cloning of classical
states, or probability distribution in asymptotic setting. We study the
quality of the approximate ($n,rn$)-clone, with $n$ being very large and $r$
being constant.

The result turns out to be $\left\Vert \mathrm{N}\left(  0,r\mathbf{1}\right)
-\mathrm{N}\left(  0,\mathbf{1}\right)  \right\Vert _{1}$, where
$\mathrm{N}\left(  \mu,\Sigma\right)  $ is the Gaussian distribution with mean
$\mu$ and covariance $\Sigma$. Notablly, this value does not depend on the the
family of porbability distributions to be cloned.

The key of the argument is use of local asymptotic normality: If the curve
$\theta\rightarrow P_{\theta}$ is sufficiently smooth in $\theta$, then, the
behavior of $P_{\theta^{\prime}}^{\otimes n}$ where $\theta^{\prime}%
-\theta=o\left(  \sqrt{1/n}\right)  $, is approximated by Gaussian shift.
Using this, we reduce the general case to Gaussian shift model.

\end{abstract}

\section{Introduction}

Cloning, or approximate cloning, is one of basic operations in quantum
information processing. It is related to optimal eavesdropping of quantum key
distribution, and also to optimal estimation efficiency. The quality of the
approximate clone, thus, has been studied extensively \cite{Scrani}.

In this paper, we deal with cloning of classical states, or probability
distributions in asymptotic setting. The study of (approximate) cloning of
classical states had started even earlier than the proposal of no-cloning
theorem, to give a measure of information contained in \ additional
observations : they studied the quality of approximate $(n+r)$-copies made
from $n$-copies ($(n,n+r)$-clone, hereafter), with $n$ being very large and
$r$ being constant \cite{Helgeland}\cite{LeCam}\cite{Mammen}.

This paper explores another direction: we study the quality of the approximate
$\left(  n,rn\right)  $-clone with $n$ being very large and $r$ being
constant, since its extension to quantum system seems to be easier.

In the argument, we make full use of local asymptotic normality: If the curve
$\theta\rightarrow P_{\theta}$ is sufficiently smooth, then, the family
$\left\{  P_{\theta+hn^{-1/2}}^{\otimes n}\right\}  _{h\in%
\mathbb{R}
^{m}}$ is approximated by Gaussian shift $\left\{  \mathrm{N}\left(
h,J_{\theta}^{-1}\right)  \right\}  _{h\in%
\mathbb{R}
^{m}}$, where $J_{\theta}$ is the Fisher information matrix of $\left\{
P_{\theta}\right\}  _{\theta\in\Theta}$ at $\theta$. Using this fact , we
reduce the general case to the Gaussian shift model . More concretely, letting
$D_{r,\Sigma}$ be the loss of optimal $\left(  1,r\right)  $-cloner of the
Gaussian shift $\left\{  \mathrm{N}\left(  h,\Sigma\right)  \right\}  _{h\in%
\mathbb{R}
^{m}}$, we show
\begin{equation}
\sup_{a\geq0}\varliminf_{n\rightarrow\infty}\inf_{\Lambda}\sup_{\left\Vert
\theta^{\prime}-\theta\right\Vert \leq an^{-1/2}}\left\Vert \Lambda\left(
P_{\theta^{\prime}}^{n}\right)  -P_{\theta^{\prime}}^{rn}\right\Vert _{1}\geq
D_{r,J_{\theta}^{-1}},
\end{equation}
where $\Lambda$ moves over all the Markov maps. In other words, the loss of
the optimal asymptotic $\left(  n,nr\right)  $-cloner is asymptotically
lower-bounded by $D_{r,J_{\theta}^{-1}}$, at each $\theta\in\Theta$. This loss
turns out to be achievable: we construct a cloner $\Lambda_{\delta
,\varepsilon}^{n,r}$ with
\begin{equation}
\lim_{\delta\downarrow0}\lim_{\varepsilon\downarrow0}\lim_{n\rightarrow\infty
}\left\Vert \Lambda_{\delta,\varepsilon}^{n,r}\left(  P_{\theta}^{n}\right)
-P_{\theta}^{rn}\right\Vert _{1}=D_{r,J_{\theta}^{-1}}.
\end{equation}
Also, we find more explicit expression of $D_{r,\Sigma}$, which is
\[
D_{r,\Sigma}=\left\Vert \mathrm{N}\left(  0,r\mathbf{1}\right)  -\mathrm{N}%
\left(  0,\mathbf{1}\right)  \right\Vert _{1}.
\]

It is notable that $D_{r,\Sigma}$ does not depend on $\Sigma$. This means that
$D_{r,J_{\theta}^{-1}}$, the smallest asymptotic loss of $\left(  n,rn\right)
$-cloner, does not depend on the family $\left\{  P_{\theta}\right\}
_{\theta\in\Theta}$ to be cloned.

Since there is a (finite dimensional) quantum version of local asymptotic
normality, this argument may be extended to finite dimensional quantum case.

The paper is organized as follows. First, we give the optimal approximate
cloners for Gaussian shift families, and find some properties of them. Second,
we state local asymptotic normality of smooth family of probability
distributions, and its uniform version. Finally, we give asymptotic analysis
of approximate $\left(  n,rn\right)  $-clone of smooth families. The paper is
concluded by discussions.

\section{Gaussian shift family}

\subsection{Reduction of cloning to amplification}

The contents of the subsection is well-known, but added for the sake of completion.

Consider the Gaussian shift family $\left\{  \mathrm{N}\left(  h,\Sigma
\right)  \right\}  _{h\in%
\mathbb{R}
^{m}}$. Then, the problem of optimum approximate $\left(  1,r\right)  $-clone,
or finding a map achieving
\[
C_{r,\Sigma}:=\inf_{\Lambda:\text{Markov}}\sup_{h\in%
\mathbb{R}
^{m}}\left\Vert \Lambda\left(  \mathrm{N}\left(  h,\Sigma\right)  \right)
-\mathrm{N}\left(  h,\Sigma\right)  ^{\otimes r}\right\Vert _{1}%
\]
is equivalent to finding the optimum $r$-amplifier, or a Markov map achieving
\begin{equation}
D_{r,\Sigma}:=\inf_{\Lambda:\text{Markov}}\sup_{h\in%
\mathbb{R}
^{m}}\left\Vert \Lambda\left(  \mathrm{N}\left(  h,\Sigma\right)  \right)
-\mathrm{N}\left(  \sqrt{r}h,\Sigma\right)  \right\Vert _{1}. \label{amp-def}%
\end{equation}

To see this, let $X_{1},\cdots,X_{r}\sim$ $\mathrm{N}\left(  h,\Sigma\right)
$, and
\[
X_{i}^{\prime}=\sum_{j=1}^{r}O_{i,j}X_{j}%
\]
where $O$ is an orthogonal matrix with $O_{1,1}=O_{1,2}=\cdots=O_{1,r}%
=\frac{1}{\sqrt{r}}$. Then, $X_{1}^{\prime}\sim\mathrm{N}\left(  \sqrt
{r}h,\Sigma\right)  $ and $X_{2}^{\prime},\cdots,X_{r}^{\prime}\sim
\mathrm{N}\left(  0,\Sigma\right)  $.

Therefore, if
\[
\sup_{h\in%
\mathbb{R}
^{m}}\left\Vert \Lambda_{0}\left(  \mathrm{N}\left(  h,\Sigma\right)  \right)
-\mathrm{N}\left(  h,\Sigma\right)  ^{\otimes r}\right\Vert _{1}=C_{r,\Sigma
}+\varepsilon\text{,}%
\]
then%
\[
\sup_{h\in%
\mathbb{R}
^{m}}\left\Vert \Psi\circ\Lambda_{0}\left(  \mathrm{N}\left(  h,\Sigma\right)
\right)  -\mathrm{N}\left(  \sqrt{r}h,\Sigma\right)  \right\Vert _{1}\leq
C_{r,\Sigma}+\varepsilon,,
\]
where $\Psi$ is a Markov map corresponding to application of $O$ followed by
restriction to the first variable. Hence,
\[
C_{r,\Sigma}\geq D_{r,\Sigma}.
\]

On the other hand, let $\Psi^{\prime}$ be a Markov map corresponding to the
map
\[
X\rightarrow\left(  X,X_{2}^{\prime},\cdots,X_{r}^{\prime}\right)
,\,X_{2}^{\prime},\cdots,X_{r}^{\prime}\sim\mathrm{N}\left(  0,\Sigma\right)
\]
followed by $O^{-1}$. If
\[
\sup_{h\in%
\mathbb{R}
^{m}}\left\Vert \Lambda_{1}\left(  \mathrm{N}\left(  h,\Sigma\right)  \right)
-\mathrm{N}\left(  \sqrt{r}h,\Sigma\right)  \right\Vert _{1}=D_{r,\Sigma
}+\varepsilon,
\]
then%
\[
\sup_{h\in%
\mathbb{R}
^{m}}\left\Vert \Psi^{\prime}\circ\Lambda_{1}\left(  \mathrm{N}\left(
h,\Sigma\right)  \right)  -\mathrm{N}\left(  h,\Sigma\right)  ^{\otimes
r}\right\Vert _{1}\leq D_{r,\Sigma}+\varepsilon.
\]
Hence,
\[
C_{r,\Sigma}\leq D_{r,\Sigma}.
\]
After all, we have $C_{r,\Sigma}=D_{r,\Sigma}.$

\subsection{Amplifier for Gaussian shift families}

In this subsection, we find the optimum $r$-amplifier ($r\geq1$) and its loss
$D_{r,\Sigma}=C_{r,\Sigma}$ for the Gaussian shift family $\left\{
\mathrm{N}\left(  h,\Sigma\right)  \right\}  _{h\in%
\mathbb{R}
^{m}}$.

Observe first that
\[
\Psi_{\sqrt{r}}\left(  \mathrm{N}\left(  h,\Sigma\right)  \right)
=\mathrm{N}\left(  \sqrt{r}h,r\Sigma\right)  ,\,\,\Psi_{r^{-1/2}}\left(
\mathrm{N}\left(  \sqrt{r}h,r\Sigma\right)  \right)  =\mathrm{N}\left(
h,\Sigma\right)  .
\]
where $\Psi_{a}$ describes the Markov map corresponding to scale change.
Hence,
\begin{align*}
D_{r,\Sigma}  &  \leq\inf_{\Lambda}\sup_{h\in%
\mathbb{R}
^{m}}\left\Vert \Lambda\circ\Psi_{\sqrt{r}}\left(  \mathrm{N}\left(
h,\Sigma\right)  \right)  -\mathrm{N}\left(  \sqrt{r}h,\Sigma\right)
\right\Vert _{1}\\
&  =\inf_{\Lambda}\sup_{h\in%
\mathbb{R}
^{m}}\left\Vert \Lambda\left(  \mathrm{N}\left(  \sqrt{r}h,r\Sigma\right)
\right)  -\mathrm{N}\left(  \sqrt{r}h,\Sigma\right)  \right\Vert _{1}%
\end{align*}
and
\begin{align*}
D_{r,\Sigma}  &  =\inf_{\Lambda}\sup_{h\in%
\mathbb{R}
^{m}}\left\Vert \Lambda\circ\Psi_{r^{-1/2}}\left(  \mathrm{N}\left(  \sqrt
{r}h,r\Sigma\right)  \right)  -\mathrm{N}\left(  \sqrt{r}h,\Sigma\right)
\right\Vert _{1}\\
&  \geq\inf_{\Lambda}\sup_{h\in%
\mathbb{R}
^{m}}\left\Vert \Lambda\left(  \mathrm{N}\left(  \sqrt{r}h,r\Sigma\right)
\right)  -\mathrm{N}\left(  \sqrt{r}h,\Sigma\right)  \right\Vert _{1}.
\end{align*}
Thus,
\begin{equation}
D_{r,\Sigma}=\inf_{\Lambda}\sup_{h\in%
\mathbb{R}
^{m}}\left\Vert \Lambda\left(  \mathrm{N}\left(  \sqrt{r}h,r\Sigma\right)
\right)  -\mathrm{N}\left(  \sqrt{r}h,\Sigma\right)  \right\Vert _{1},
\label{amp-1}%
\end{equation}
and $\Lambda_{\mathrm{amp}}^{r}$ achieving (\ref{amp-def}) and $\Lambda^{r}$
achieving (\ref{amp-1}) are, if exists, related by
\[
\Lambda_{\mathrm{amp}}^{r}=\Lambda^{r}\circ\Psi_{\sqrt{r}}.
\]

Now, we refer to Theorem\thinspace3 of \cite{Torgersen:trans}: applying to our
case, it says that
\begin{align}
D_{r,\Sigma}  &  =\sup_{f:\sup_{x}\left\vert f\left(  x\right)  \right\vert
\leq1}\left\{  \int f\left(  y\right)  p_{0,\Sigma}\left(  y\right)
\mathrm{d}y-\sup_{x}\int f\left(  y+\sqrt{r}x\right)  p_{0,r\Sigma}\left(
y\right)  \mathrm{d}y\right\} \nonumber\\
&  =\sup_{f:\sup_{x}\left\vert f\left(  x\right)  \right\vert \leq1}\inf
_{x}\left\{  \int f\left(  y\right)  \left\{  p_{0,\Sigma}\left(  y\right)
-p_{x,r\Sigma}\left(  y\right)  \right\}  \mathrm{d}y\right\} \nonumber\\
&  =\sup_{f:\sup_{x}\left\vert f\left(  x\right)  \right\vert \leq1}\inf
_{x}\left\{  \int f\left(  y\right)  \left\{  p_{0,\mathbf{1}}\left(
y\right)  -p_{x,r\mathbf{1}}\left(  y\right)  \right\}  \mathrm{d}y\right\}  ,
\label{D=}%
\end{align}
where $p_{x,\Sigma}$ is probability density function of $\mathrm{N}\left(
x,\Sigma\right)  $.

The right most side of (\ref{D=}) is evaluated as follows. Observe
\begin{align}
D_{r,\Sigma}  &  \leq\inf_{x}\left\Vert p_{0,\mathbf{1}}-p_{x,r\mathbf{1,}%
}\right\Vert _{1}\nonumber\\
&  =\left\Vert p_{0,\mathbf{1}}-p_{0,r\mathbf{1}}\right\Vert _{1}.
\label{inf|p-p|}%
\end{align}
(The proof of (\ref{inf|p-p|}) is in the appendix.) On the other hand, define
$B_{r}:=\left\{  y\,;p_{\mathbf{1}}\left(  y\right)  \geq p_{r\mathbf{1}%
}\left(  y\right)  \right\}  $, which is a ball centered at origin. Then,%
\begin{align}
D_{r,\Sigma}  &  \geq\inf_{x}\left\{  \int\left(  2I_{B_{r}}\left(  y\right)
-1\right)  \left\{  p_{0,\mathbf{1}}\left(  y\right)  -p_{x,r\mathbf{1}%
}\left(  y\right)  \right\}  \mathrm{d}y\right\} \nonumber\\
&  =\int\left(  2I_{B_{r}}\left(  y\right)  -1\right)  p_{0,\mathbf{1}}\left(
y\right)  \mathrm{d}y-\sup_{x}\int\left(  2I_{B_{r}}\left(  y\right)
-1\right)  p_{x,r\mathbf{1}}\left(  y\right)  \mathrm{d}y\nonumber\\
&  =\int\left(  2I_{B_{r}}\left(  y\right)  -1\right)  p_{0,\mathbf{1}}\left(
y\right)  \mathrm{d}y-\int\left(  2I_{B_{r}}\left(  y\right)  -1\right)
p_{0,r\mathbf{1}}\left(  y\right)  \mathrm{d}y\nonumber\\
&  =\left\Vert p_{0,\mathbf{1}}-p_{0,r\mathbf{1}}\right\Vert _{1}.
\label{infint}%
\end{align}
(N.B. in the case of $r<1$, $\sup_{x}\int\left(  2I_{B_{r}}\left(  y\right)
-1\right)  p_{x,r\mathbf{1}}\left(  y\right)  \mathrm{d}y$ is achieved as
$\left\Vert x\right\Vert \rightarrow\infty$.)

\ After all, we have, if $r\geq1$,
\begin{equation}
D_{r,\Sigma}=\left\Vert p_{0,\mathbf{1}}-p_{0,r\mathbf{1}}\right\Vert
_{1}=\left\Vert \mathrm{N}\left(  0,\mathbf{1}\right)  -\mathrm{N}\left(
0,r\mathbf{1}\right)  \right\Vert _{1}. \label{gauss-opt-loss}%
\end{equation}
Obviously, corresponding $\Lambda^{r}$ is the identity map. Thus,
\begin{equation}
\Lambda_{\mathrm{amp}}^{r}=\Psi_{\sqrt{r}}. \label{gauss-shift-opt}%
\end{equation}

\subsection{Bounded shifts}

Define
\[
D_{r,\Sigma,a}:=\inf_{\Lambda}\sup_{\left\Vert h\right\Vert \leq a}\left\Vert
\Lambda\left(  \mathrm{N}\left(  h,\Sigma\right)  \right)  -\mathrm{N}\left(
h,\sqrt{r}\Sigma\right)  \right\Vert _{1}.
\]
Then, if $a^{\prime}\geq a$ and \
\[
\sup_{\left\Vert h\right\Vert \leq a^{\prime}}\left\Vert \Lambda\left(
\mathrm{N}\left(  h,\Sigma\right)  \right)  -\mathrm{N}\left(  h,\sqrt
{r}\Sigma\right)  \right\Vert _{1}=D_{r,\Sigma,a^{\prime}}+\varepsilon,
\]
then
\[
\sup_{\left\Vert h\right\Vert \leq a}\left\Vert \Lambda\left(  \mathrm{N}%
\left(  h,\Sigma\right)  \right)  -\mathrm{N}\left(  h,\sqrt{r}\Sigma\right)
\right\Vert _{1}\leq D_{r,\Sigma,a^{\prime}}+\varepsilon.
\]
Since $\varepsilon>0$ can be arbitrary, therefore,
\[
D_{r,\Sigma,a}\leq D_{r,\Sigma,a^{\prime}}.
\]
Hence, since \ $D_{r,\Sigma,a}\leq2$, \ $\lim_{a\rightarrow\infty}%
D_{r,\Sigma,a}$ exists.

\begin{lemma}%
\[
\lim_{a\rightarrow\infty}D_{r,\Sigma,a}=D_{r,\Sigma}.
\]

\end{lemma}

\begin{proof}
Let us consider a decision problem taking values in $\left[  -1,1\right]  ^{%
\mathbb{R}
^{m}}$. Let $\rho$ be a Markov kernel from $%
\mathbb{R}
^{m}$ to $\left[  -1,1\right]  ^{%
\mathbb{R}
^{m}}$, and $F\left(  h,\cdot\right)  :%
\mathbb{R}
^{m}\rightarrow\left[  -1,1\right]  $ be a lower continuous function. Also, we
define $\mathcal{P}_{a}$ be the set of probability distributions over
$\left\{  x;\left\Vert x\right\Vert \leq a\right\}  $ with finite support.
Then, we define, for $\pi\in\mathcal{P}_{a}$,
\[
R_{\pi}\left(  \Sigma,F,\rho\right)  :=\int\int F\left(  h,a\right)
\rho\left(  \mathrm{d}a,x\right)  p_{h,r\Sigma}\left(  x\right)
\mathrm{d}x\mathrm{d}\pi\left(  h\right)  .
\]
Due to the randomization criteria (Theorem\thinspace1.10 of \cite{Shiryaev},
Theorem 55.9 of \cite{Strasser}),
\[
D_{r,\Sigma}=\sup_{\pi\in\mathcal{P}_{\infty}}\sup_{F}\left\{  \inf_{\rho
}R_{\pi}\left(  r\Sigma,F,\rho\right)  -\inf_{\rho}R_{\pi}\left(
\Sigma,F,\rho\right)  \right\}  ,
\]
and
\[
D_{r,\Sigma,a}=\sup_{\pi\in\mathcal{P}_{a}}\sup_{F}\left\{  \inf_{\rho}R_{\pi
}\left(  r\Sigma,F,\rho\right)  -\inf_{\rho}R_{\pi}\left(  \Sigma
,F,\rho\right)  \right\}  .
\]
Comparing the right hand sides of them,
\[
D_{r,\Sigma}=\sup_{a\geq0}D_{r,\Sigma,a}=\lim_{a\rightarrow\infty}%
D_{r,\Sigma,a}.
\]

\end{proof}

\section{Smooth family}

\subsection{Settings and description of results}

Consider a family of probability distributions $\left\{  P_{\theta};\theta
\in\Theta\right\}  $ over the measurable space $\left(  \Omega,\mathcal{X}%
\right)  $, where $\Theta$ is an open region in $%
\mathbb{R}
^{m}$, $\Omega$ is a Polish space (a separable completely metrizable
topological space, e.g. $%
\mathbb{R}
^{k}$, $%
\mathbb{Z}
^{k}$, etc. ) , and $P_{\theta}$ has density $p_{\theta}$ with respect to a
measure $\mu$. Define $P_{\theta}^{n}:=P_{\theta}^{\otimes n}$, $p_{\theta
}^{n}:=p_{\theta}^{\otimes n}$, $\Omega^{n}:=\Omega^{\times n}$,
$\mathcal{X}^{n}:=\mathcal{X}^{\otimes n}$, and%

\[
Z_{\theta,h}^{n}:=\frac{p_{\theta+hn^{-1/2}}^{n}}{p_{\theta}^{n}}.
\]
Also, $\mathrm{E}_{\theta}$ and $\mathrm{E}_{\theta}^{n}$ refers to
expectation with respect to $P_{\theta}$ or $P_{\theta}^{n}$, respectively.
$W_{\theta,\kappa}$ ($\kappa=1,\cdots,n$) are the random variables with
$W_{\theta,\kappa}\sim P_{\theta}$, and define $W_{\theta}^{n}:=\left(
W_{\theta,1},\cdots,W_{\theta,n}\right)  $, which obeys $P_{\theta}^{n}$.

Under this setting, we investigate the quality of $\left(  n,nr\right)
$-clone of $\left\{  P_{\theta};\theta\in\Theta\right\}  $. More specifically,
we show
\begin{equation}
\sup_{a\geq0}\varliminf_{n\rightarrow\infty}\inf_{\Lambda^{n,r}\text{:Markov}%
}\sup_{\left\Vert \theta^{\prime}-\theta\right\Vert \leq an^{-1/2}}\left\Vert
\Lambda^{n,r}\left(  P_{\theta^{\prime}}^{n}\right)  -P_{\theta^{\prime}}%
^{rn}\right\Vert _{1}\geq D_{r,J_{\theta}^{-1},\infty}=D_{r,J_{\theta}^{-1}},
\label{l-min-max}%
\end{equation}
which means the loss of the optimal asymptotic $\left(  n,nr\right)  $-cloner
is lower bounded by $D_{r,J_{\theta}^{-1}}$, at each $\theta\in\Theta$. Also,
we show this loss is achievable: we construct a cloner $\Lambda_{\delta
,\varepsilon}^{n,r}$ with
\begin{equation}
\lim_{\delta\downarrow0}\lim_{\varepsilon\downarrow0}\lim_{n\rightarrow\infty
}\left\Vert \Lambda_{\delta,\varepsilon}^{n,r}\left(  P_{\theta}^{n}\right)
-P_{\theta}^{rn}\right\Vert _{1}=D_{r,J_{\theta}^{-1}}. \label{achieve}%
\end{equation}

\subsection{Local asymptotic normality and its uniform version}

The map $\theta\rightarrow p_{\theta}$ is \textit{\ differentiable in
quadratic mean}, if
\begin{equation}
\lim_{h\rightarrow0}\frac{1}{\left\Vert h\right\Vert ^{2}}\int\left(
\sqrt{p_{\theta+h}}-\sqrt{p_{\theta}}-\frac{h^{T}}{2}\ell_{\theta}%
\sqrt{p_{\theta}}\right)  ^{2}\mathrm{d}\mu=0,\,\,\forall\theta\in\Theta.
\label{diff-quadratic-1}%
\end{equation}
If the map $\theta\rightarrow\ell_{\theta}$ is continuous, we say
$\theta\rightarrow p_{\theta}$ is \textit{continuously differentiable in
quadratic mean}.

We define, with $\omega^{n}\in\Omega^{n}$ and $\omega_{\kappa}\in\Omega$,
\[
\ell_{\theta}^{n}\left(  \omega^{n}\right)  :=\frac{1}{\sqrt{n}}\sum
_{\kappa=1}^{n}\ell_{\theta}\left(  \omega_{\kappa}\right)  ,
\]
$J_{\theta}:=\left[  \mathrm{E}_{\theta}\ell_{\theta,i}\ell_{\theta,j}\right]
$, and \
\[
Z_{\theta,h}\left(  x\right)  :=\exp\left(  h^{T}x-\frac{1}{2}h^{T}J_{\theta
}h\right)  .
\]

The following Lemma is recasting of Remark 1 of \cite{Ibragimov} and
Theorem\thinspace7.2 of \cite{vandervaart}.

\begin{lemma}
\label{lem:likelihood-expansion} Suppose $\Theta$ is an open region in $%
\mathbb{R}
^{m}$ and $\theta\rightarrow p_{\theta}$ is continuously\textit{
differentiable in quadratic mean. }Then, $\mathrm{E}_{\theta}\ell_{\theta}=0$,
and, for any compact set $K\subset\Theta$ and $K^{\prime}\subset%
\mathbb{R}
^{m}$,
\[
\lim_{n\rightarrow\infty}\sup_{h\in K^{\prime}}\sup_{\theta\in K}P_{\theta
}^{n}\left\{  \left\vert \ln Z_{\theta,h}^{n}-\ln Z_{\theta,h}\left(
\ell_{\theta}^{n}\right)  \right\vert >\varepsilon\right\}  =0,\,\forall
\varepsilon>0.
\]

\end{lemma}

The following Lemma is recasting of Remark 1 of Theorem\thinspace7.2 of
\cite{vandervaart}.

\begin{lemma}
\label{lem:likelihood-expansion-1}Suppose $\Theta$ is an open region in $%
\mathbb{R}
^{m}$ and $\theta\rightarrow p_{\theta}$ is \textit{\ differentiable in
quadratic mean. }Then, $\mathrm{E}_{\theta}\ell_{\theta}=0$, and, for any
compact set $K^{\prime}\subset%
\mathbb{R}
^{m}$,
\[
\lim_{n\rightarrow\infty}\sup_{h\in K^{\prime}}P_{\theta}^{n}\left\{
\left\vert \ln Z_{\theta,h}^{n}-\ln Z_{\theta,h}\left(  \ell_{\theta}%
^{n}\right)  \right\vert >\varepsilon\right\}  =0,\,\forall\varepsilon>0.
\]

\end{lemma}

In addition, we assume the following conditions:%
\begin{align}
&  J_{\theta}\text{ is continuous in }\theta\text{,}\label{J-cont}\\
&  \inf_{\theta\in\Theta}\alpha_{\theta}>0,\,\, \label{J-nondeg}%
\end{align}
where $\alpha_{\theta}$ is the minimum eigenvalue of $J_{\theta}$, and \
\begin{equation}
\sup_{\theta\in K}\mathrm{E}_{\theta}e^{h^{T}\ell_{\theta}}<\infty
,\,\text{\ }\forall h\in%
\mathbb{R}
^{m},\text{ for any compact set }K\subset\Theta\text{.} \label{EZ<infty}%
\end{equation}
Observe that
\begin{align*}
\mathrm{E}_{\theta}\left(  h^{T}\ell_{\theta}\right)  ^{2k}  &  \leq\left(
2k\right)  !\left\Vert h\right\Vert ^{2k}\mathrm{E}_{\theta}\cosh\left(
e^{T}\ell_{\theta}\right)  ,\\
\mathrm{E}_{\theta}\left\vert h^{T}\ell_{\theta}\right\vert ^{2k-1}  &
\leq\left\Vert h\right\Vert ^{2k-1}\left\{  1+\mathrm{E}_{\theta}\left(
e^{T}\ell_{\theta}\right)  ^{2k}\right\}  \leq\left\Vert h\right\Vert
^{2k-1}\left\{  1+\left(  2k\right)  !\mathrm{E}_{\theta}\cosh\left(
e^{T}\ell_{\theta}\right)  \right\}  ,
\end{align*}
where $e=h/\left\Vert h\right\Vert $, implying
\begin{equation}
\sup_{\theta\in K}\mathrm{E}_{\theta}\left\vert h^{T}\ell_{\theta}\right\vert
^{k}<\infty,\,\text{\ }\forall h\in%
\mathbb{R}
^{m},\text{ for any compact set }K\subset\Theta\text{.} \label{l-moment<infty}%
\end{equation}

Also, one can show that, for any compact set $K\subset\Theta$ and $K^{\prime
}\subset%
\mathbb{R}
^{m}$,
\begin{equation}
\sup_{n\geq n_{K,K^{\prime}}}\mathrm{E}_{\theta}e^{h^{T}\ell_{\theta}^{n}}\leq
e^{h^{T}J_{\theta}h},\,\forall\theta\in\Theta,\forall h\in K^{\prime},\exists
n_{K,K^{\prime}}\, \label{EZ<finite}%
\end{equation}
The proof of (\ref{EZ<finite}) is as follows. Observe, since $\mathrm{E}%
_{\theta}\ell_{\theta}=0$ due to Lemma\thinspace
\ref{lem:likelihood-expansion-1},
\begin{align}
\mathrm{E}_{\theta}^{n}e^{h^{T}\ell_{\theta}^{n}}  &  =\left(  \mathrm{E}%
_{\theta}e^{-\frac{h^{T}}{\sqrt{n}}\ell_{\theta}}\right)  ^{n}\nonumber\\
&  =\left(  1+\frac{h^{T}J_{\theta}h}{2n}+f_{\mathrm{rem}}\left(
\theta,h,n\right)  \right)  ^{n}, \label{exp-rem}%
\end{align}
where%
\begin{align}
&  \left\vert f_{\mathrm{rem}}\left(  \theta,h,n\right)  \right\vert
\nonumber\\
&  \leq\sum_{k=3}^{\infty}\frac{1}{k!}\left(  \frac{\left\Vert h\right\Vert
}{\sqrt{n}}\right)  ^{k}\mathrm{E}_{\theta}\left\vert e^{T}\ell_{\theta
}\right\vert ^{k}\nonumber\\
&  \leq\frac{1}{2}\sum_{k\geq3,k\text{:even}}^{\infty}\left(  \frac{\left\Vert
h\right\Vert }{\sqrt{n}}\right)  ^{k}\mathrm{E}_{\theta}\cosh e^{T}%
\ell_{\theta}\nonumber\\
&  +\sum_{k\geq3,k\text{:odd}}^{\infty}\left\{  \frac{1}{k!}\left(
\frac{\left\Vert h\right\Vert }{\sqrt{n}}\right)  ^{k}+\frac{\left(
k+1\right)  !}{k!}\left(  \frac{\left\Vert h\right\Vert }{\sqrt{n}}\right)
^{k}\mathrm{E}_{\theta}\cosh e^{T}\ell_{\theta}\right\} \nonumber\\
&  \leq\sum_{k\geq3}^{\infty}\left(  k+1\right)  \left(  \frac{\left\Vert
h\right\Vert }{\sqrt{n}}\right)  ^{k}\left\{  \mathrm{E}_{\theta}\cosh
e^{T}\ell_{\theta}+1\right\} \nonumber\\
&  =\left(  \frac{\left\Vert h\right\Vert }{\sqrt{n}}\right)  ^{3}%
\frac{4-5\left\Vert h\right\Vert /\sqrt{n}}{1-\left\Vert h\right\Vert
/\sqrt{n}}\left(  \mathrm{E}_{\theta}\cosh e^{T}\ell_{\theta}+1\right)  .
\label{rem-upper}%
\end{align}
Therefore, for each compact set $K\subset\Theta$ and $K^{\prime}\subset%
\mathbb{R}
^{m}$, there is $n_{K,K^{\prime}}$ such that
\[
\mathrm{E}_{\theta}^{n}e^{h^{T}\ell_{\theta}^{n}}\leq\left(  1+\frac
{h^{T}J_{\theta}h}{n}\right)  ^{n}\leq e^{h^{T}J_{\theta}h},\,\forall n\geq
n_{K,K^{\prime}}.
\]
Hence, we have (\ref{EZ<finite}).

Also, we use the following identity :%
\begin{equation}
\lim_{a\rightarrow\infty}\sup_{n\geq n_{K,K^{\prime}}}\sup_{h\in K^{\prime}%
}\sup_{\theta\in K}\mathrm{E}_{\theta}^{n}\left[  e^{h^{T}\ell_{\theta}^{n}%
}:e^{h^{T}\ell_{\theta}^{n}}\geq a\right]  =0, \label{E[e>a]}%
\end{equation}
which is proved as follows.%

\begin{align*}
&  \lim_{a\rightarrow\infty}\sup_{n\geq n_{K,K^{\prime}}}\sup_{h\in K^{\prime
}}\sup_{\theta\in K}\mathrm{E}_{\theta}^{n}\left[  e^{h^{T}\ell_{\theta}^{n}%
}:e^{h^{T}\ell_{\theta}^{n}}\geq a\right] \\
&  \leq\lim_{a\rightarrow\infty}\sup_{n\geq n_{K,K^{\prime}}}\sup_{h\in
K^{\prime}}\sup_{\theta\in K}\sqrt{\mathrm{E}_{\theta}^{n}\left[
e^{2h^{T}\ell_{\theta}^{n}}\right]  P_{\theta}^{n}\left\{  e^{h^{T}%
\ell_{\theta}^{n}}\geq a\right\}  }\\
&  \leq\lim_{a\rightarrow\infty}\sup_{n\geq n_{K,K^{\prime}}}\sup_{h\in
K^{\prime}}\sup_{\theta\in K}e^{2h^{T}J_{\theta}h}\sqrt{P_{\theta}^{n}\left\{
e^{h^{T}\ell_{\theta}^{n}}\geq a\right\}  }\\
&  \leq\lim_{a\rightarrow\infty}\sup_{n\geq n_{K,K^{\prime}}}\sup_{h\in
K^{\prime}}\sup_{\theta\in K}e^{2h^{T}J_{\theta}h}\sqrt{\frac{1}{a}%
\mathrm{E}_{\theta}^{n}e^{h^{T}\ell_{\theta}^{n}}}\\
&  \leq\lim_{a\rightarrow\infty}\frac{1}{a}\sup_{h\in K^{\prime}}\sup
_{\theta\in K}e^{\frac{5}{2}h^{T}J_{\theta}h}=0.
\end{align*}

\begin{lemma}
\label{lem:converge}Suppose random variables $X_{n,t}$ , and $Y_{n,t}$,
$n\geq1$, $t\in T$, taking values in $%
\mathbb{R}
^{k,}$ satisfies
\begin{equation}
\lim_{n\rightarrow\infty}\sup_{t\in T}\Pr\left\{  \left\Vert X_{n,t}%
-Y_{n,t}\right\Vert >\varepsilon\right\}  =0, \label{prob-converge}%
\end{equation}
Let $f$ be a continuously differentiable function from $%
\mathbb{R}
^{k}$ to $%
\mathbb{R}
$ such that,
\begin{equation}
\sup_{x:f\left(  x\right)  \leq a}\left\Vert \nabla_{x}f\left(  x\right)
\right\Vert <\infty,\,\, \label{f-1}%
\end{equation}
and
\begin{align}
\lim_{a\rightarrow\infty}\lim_{n\rightarrow\infty}\sup_{t\in T}\mathrm{E}%
\left[  f\left(  X_{n,t}\right)  :f\left(  X_{n,t}\right)  >a\right]   &
<\infty,\label{f-2}\\
\,\lim_{a\rightarrow\infty}\lim_{n\rightarrow\infty}\sup_{t\in T}%
\mathrm{E}\left[  f\left(  X_{n,t}\right)  :f\left(  Y_{n,t}\right)
>a\right]   &  <\infty. \label{f-3}%
\end{align}
Then,
\[
\lim_{n\rightarrow\infty}\sup_{t\in T}\left\vert \mathrm{E}f\left(
X_{n,t}\right)  -\mathrm{E}f\left(  Y_{n,t}\right)  \right\vert =0
\]

\end{lemma}

\begin{proof}
Define%
\[
f^{a}\left(  x\right)  :=f\left(  x\right)  \wedge a.
\]
Then,
\begin{align}
&  \lim_{n\rightarrow\infty}\sup_{t\in T}\left\vert \mathrm{E}f\left(
X_{n,t}\right)  -\mathrm{E}f\left(  Y_{n,t}\right)  \right\vert \nonumber\\
&  \leq\lim_{n\rightarrow\infty}\sup_{t\in T}\left\vert \mathrm{E}f^{a}\left(
X_{n,t}\right)  -\mathrm{E}f^{a}\left(  X_{t}\right)  \right\vert \nonumber\\
&  +\lim_{n\rightarrow\infty}\sup_{t\in T}\left\vert \mathrm{E}\left[
f\left(  X_{n,t}\right)  :f\left(  X_{n,t}\right)  >a\right]  \right\vert
\nonumber\\
&  +\lim_{n\rightarrow\infty}\sup_{t\in T}\left\vert \mathrm{E}\left[
f\left(  Y_{n,t}\right)  :f\left(  Y_{n,t}\right)  >a\right]  \right\vert .
\label{E-E}%
\end{align}

The first term of the right hand side is evaluated as follows.%
\[
\left\vert f^{a}\left(  X_{n,t}\right)  -f^{a}\left(  Y_{n,t}\right)
\right\vert \leq C\left\Vert X_{n,t}-Y_{n,t}\right\Vert ,\forall t\in T
\]
where%
\[
C\leq\sup_{x:f\left(  x\right)  \leq a}\left\Vert \nabla_{x}f\left(  x\right)
\right\Vert <\infty.
\]
Therefore,
\begin{align*}
\lim_{n\rightarrow\infty}\sup_{t\in T}\left\vert \mathrm{E}f^{a}\left(
X_{n,t}\right)  -\mathrm{E}f^{a}\left(  Y_{n,t}\right)  \right\vert  &
\leq\varepsilon+a\times\lim_{n\rightarrow\infty}\sup_{t\in T}\Pr\left\{
\left\vert f^{a}\left(  X_{n,t}\right)  -f^{a}\left(  Y_{n,t}\right)
\right\vert >\varepsilon\right\} \\
&  =\varepsilon+a\times\lim_{n\rightarrow\infty}\sup_{t\in T}\Pr\left\{
C\left\vert X_{n,t}-Y_{n,t}\right\vert >\varepsilon\right\} \\
&  =\varepsilon.
\end{align*}
This can be made arbitrarily small, since $\varepsilon>0$ is arbitrary.

The second and the third terms of the right hand side of ( \ref{E-E}) can be
made \ arbitrarily small by taking $a$ large. Hence, we have the assertion.
\end{proof}

\begin{lemma}
\label{lem:EZ-EZl}Suppose  $\theta\rightarrow p_{\theta}$ is continuously
differentiable in quadratic mean, and (\ref{EZ<infty}) holds. Then, for any
compact set $K\subset\Theta$ and $K^{\prime}\subset%
\mathbb{R}
^{m}$,
\[
\lim_{n\rightarrow\infty}\sup_{h\in K^{\prime}}\sup_{\theta\in K}%
\mathrm{E}_{\theta}^{n}\left\vert Z_{\theta,h}^{n}-Z_{\theta,h}\left(
\ell_{\theta}^{n}\right)  \right\vert =0.
\]

\end{lemma}

\begin{proof}
We apply Lemma\thinspace\ref{lem:converge}, with $f\left(  x\right)  :=e^{x}$,
$t=\left(  \theta,h\right)  $, $X_{n,t}:=\ln Z_{\theta,h}^{n}$ and
$Y_{n,t}:=\ln Z_{\theta,h}\left(  \ell_{\theta}^{n}\right)  =h^{T}\ell
_{\theta}^{n}-\frac{1}{2}h^{T}J_{\theta}h$. Then, the premises
(\ref{prob-converge}) and (\ref{f-1}) are obviously satisfied.

Due to (\ref{E[e>a]}), (\ref{f-2}) is satisfied:
\[
\lim_{n\rightarrow\infty}\sup_{h\in K^{\prime}}\sup_{\theta\in K}%
\mathrm{E}_{\theta}^{n}\left[  Z_{\theta,h}\left(  \ell_{\theta}^{n}\right)
:Z_{\theta,h}\left(  \ell_{\theta}^{n}\right)  \geq a\right]  \rightarrow
0,a\rightarrow\infty.
\]

(\ref{f-3}) is proved as follows. Let $g_{a}\left(  x\right)  $ be a
continuous function on $%
\mathbb{R}
_{+}$ such that $g_{a}\left(  x\right)  =1$ for $x\leq a-1$ and $g_{a}\left(
x\right)  =0$ for $x\geq a$. Then,
\begin{align*}
&  \lim_{n\rightarrow\infty}\sup_{h\in K^{\prime}}\sup_{\theta\in K}%
\mathrm{E}_{\theta}^{n}\left[  Z_{\theta,h}^{n}\,:Z_{\theta,h}^{n}\geq
a\right] \\
&  \leq\lim_{n\rightarrow\infty}\sup_{h\in K^{\prime}}\sup_{\theta\in
K}\left\{  1-\mathrm{E}_{\theta}^{n}\left[  Z_{\theta,h}^{n}\,g_{a}\left(
Z_{\theta,h}^{n}\,\right)  \right]  \right\} \\
&  \leq\lim_{n\rightarrow\infty}\sup_{h\in K^{\prime}}\sup_{\theta\in
K}\left\{  1-\mathrm{E}_{\theta}^{n}\left[  Z_{\theta,h}\left(  \ell_{\theta
}^{n}\right)  \,g_{a}\left(  Z_{\theta,h}\left(  \ell_{\theta}^{n}\right)
\,\,\right)  \right]  \right\} \\
&  +\sup_{x}\left\{  \left(  x+\varepsilon\right)  g_{a}\left(  x+\varepsilon
\right)  -xg_{a}\left(  x\right)  \right\} \\
&  +a\lim_{n\rightarrow\infty}P_{\theta}^{n}\left\{  \left\vert Z_{\theta
,h}^{n}-Z_{\theta,h}\left(  \ell_{\theta}^{n}\right)  \right\vert
>\varepsilon\right\} \\
&  \leq\lim_{n\rightarrow\infty}\sup_{h\in K^{\prime}}\sup_{\theta\in
K}\mathrm{E}_{\theta}^{n}\left[  Z_{\theta,h}\left(  \ell_{\theta}^{n}\right)
\,:Z_{\theta,h}\left(  \ell_{\theta}^{n}\right)  \geq a-1\,\,\right] \\
&  +\sup_{x}\left\{  \left(  x+\varepsilon\right)  g_{a}\left(  x+\varepsilon
\right)  -xg_{a}\left(  x\right)  \right\} \\
&  +a\lim_{n\rightarrow\infty}P_{\theta}^{n}\left\{  \left\vert Z_{\theta
,h}^{n}-Z_{\theta,h}\left(  \ell_{\theta}^{n}\right)  \right\vert
>\varepsilon\right\} \\
&  =\lim_{n\rightarrow\infty}\sup_{h\in K^{\prime}}\sup_{\theta\in
K}\mathrm{E}_{\theta}^{n}\left[  Z_{\theta,h}\left(  \ell_{\theta}^{n}\right)
\,:Z_{\theta,h}\left(  \ell_{\theta}^{n}\right)  \geq a-1\,\,\right] \\
&  +\sup_{x}\left\{  \left(  x+\varepsilon\right)  g_{a}\left(  x+\varepsilon
\right)  -xg_{a}\left(  x\right)  \right\}  .
\end{align*}
Since $\varepsilon>0$ is arbitrary and $g_{a}$ is continuous,
\begin{align*}
&  \lim_{n\rightarrow\infty}\sup_{h\in K^{\prime}}\sup_{\theta\in K}%
\mathrm{E}_{\theta}^{n}\left[  Z_{\theta,h}^{n}\,:Z_{\theta,h}^{n}\geq
a\right] \\
&  \leq\lim_{n\rightarrow\infty}\sup_{h\in K^{\prime}}\sup_{\theta\in
K}\mathrm{E}_{\theta}^{n}\left[  Z_{\theta,h}\left(  \ell_{\theta}^{n}\right)
\,:Z_{\theta,h}\left(  \ell_{\theta}^{n}\right)  \geq a-1\,\,\right] \\
&  \rightarrow0,\,a\rightarrow\infty.
\end{align*}
So, we have the assertion.
\end{proof}

\begin{lemma}
\label{lem:EZ-EZl-1}Suppose  $\theta\rightarrow p_{\theta}$ is differentiable
in quadratic mean, and (\ref{EZ<infty}) holds. Then, for any compact set
$K^{\prime}\subset%
\mathbb{R}
^{m}$,
\[
\lim_{n\rightarrow\infty}\sup_{h\in K^{\prime}}\mathrm{E}_{\theta}%
^{n}\left\vert Z_{\theta,h}^{n}-Z_{\theta,h}\left(  \ell_{\theta}^{n}\right)
\right\vert =0.
\]

\end{lemma}

\begin{proof}
The proof is almost parallel with the one of Lemma\thinspace\ref{lem:EZ-EZl},
except that Lemma\thinspace\ref{lem:likelihood-expansion-1} is used instead of
Lemma\thinspace\ref{lem:likelihood-expansion} and that $\sup_{\theta\in K}$ at
each step is removed.
\end{proof}

Below, we denote by $C\left(  h,r\right)  $ the closed $m$-dimensional
hypercube which is centered at $h\in%
\mathbb{R}
^{m}$, parallel to the coordinate axis, and of edge length $2r$. Also, $2^{-k}%
\mathbb{Z}
^{m}$ is an element of $%
\mathbb{R}
^{m}$ whose coordinates are integer multiple of $2^{-k}$.

\begin{lemma}
\label{lem:cell}Let $\Theta_{0}$ be a countable subset of $\Theta$ and $c_{n}$
be a positive constant. Then, to every ordered correction $\left(  i_{1}%
,i_{2},\cdots,i_{k}\right)  $ associate a Borel set $S_{\left(  i_{1}%
,i_{2},\cdots,i_{k}\right)  }$ in $%
\mathbb{R}
^{m}$ such that%
\begin{align}
S_{\left(  i_{1},i_{2},\cdots,i_{k}\right)  }\cap S_{\left(  j_{1}%
,j_{2},\cdots,j_{k}\right)  }  &  =\emptyset,\,\,\,\,\left(  i_{1}%
,i_{2},\cdots,i_{k}\right)  \neq\left(  j_{1},j_{2},\cdots,j_{k}\right)
,\label{cell-1}\\
\text{Diameter \ of }S_{\left(  i_{1},i_{2},\cdots,i_{k}\right)  }  &
<\sqrt{m}2^{-k+2}\,\,\,\,\,\left(  k\geq1\right)  ,\label{cell-2}\\
P_{\theta}^{n}\left\{  \ell_{\theta}^{n}\in\partial S_{\left(  i_{1}%
,i_{2},\cdots,i_{k}\right)  }\right\}   &  =0,\,\,\forall\theta\in\Theta
_{0},\forall n\label{cell-3}\\
\bigcup_{j=1}^{N_{n}}S_{j}  &  \supset\left[  -c_{n},c_{n}\right]  ^{m}%
,N_{n}:=\left(  2c_{n}+1\right)  ^{m},\label{cell-4}\\
\bigcup_{j=1}^{\infty}S_{j}  &  =%
\mathbb{R}
^{m},\label{cell-5}\\
\,\bigcup_{j=1}^{5^{m}}S_{\left(  i_{1},\cdots,i_{k-1},j\right)  }  &
=S_{\left(  i_{1},\cdots,i_{k-1}\right)  }.\text{ } \label{cell-6}%
\end{align}

\end{lemma}

\begin{proof}
Since $\Theta_{0}$ is a countable set, we can choose an $r_{0}$ with
$c_{n}<r_{0}<c_{n}+\frac{1}{2}$ and
\begin{equation}
P_{\theta}^{n}\left\{  \ell_{\theta}^{n}\in C\left(  0,r_{0}\right)  \right\}
=0,\,\forall\theta\in\Theta_{0},\forall n. \label{Pr=0-0}%
\end{equation}
Also, we can choose $r_{k}$ with $2^{-k}<r_{k}<2^{-k+1}$ and
\begin{equation}
P_{\theta}^{n}\left\{  \ell_{\theta}^{n}\in C\left(  h,r_{k}\right)  \right\}
=0,\,\forall\theta\in\Theta_{0},\forall n, \label{Pr=0}%
\end{equation}
for all $h\in2^{-k+1}%
\mathbb{Z}
^{m}$.

First, we compose $S_{1}$, $S_{2}$, $\cdots$. Define $h_{j}\in%
\mathbb{Z}
^{m}$ so that $h_{1}$,$\cdots$,$h_{N_{n}}\in\left[  -c_{n},c_{n}\right]  ^{m}%
$, and that $\left\{  h_{j}\,;\,j=1,2,\cdots\right\}  =$ $%
\mathbb{Z}
^{m}$. Then, recursively define, for $j=1$,$\cdots$, $N_{n}$,
\[
S_{1}:=C\left(  0,r_{0}\right)  \cap C\left(  h_{1},r_{1}\right)
,\,\,\,S_{j}:=C\left(  0,r_{0}\right)  \cap\left\{  C\left(  h_{j}%
,r_{1}\right)  -\bigcup_{i=1}^{j-1}S_{i}\right\}  ,
\]
and, for $j\geq N_{n}+1$,
\[
S_{j}:=C\left(  h_{j},r_{1}\right)  -\bigcup_{i=1}^{j-1}S_{i}\text{.}%
\]

Since $\bigcup_{j=1}^{N_{n}}C\left(  h_{j},2^{-1}\right)  =C\left(
0,c_{n}+\frac{1}{2}\right)  $, we have $\bigcup_{j=1}^{N_{n}}C\left(
h_{j},r_{1}\right)  \supset C\left(  0,r_{0}\right)  .$ Also,
\[
\bigcup_{j=1}^{N_{n}}S_{j}=C\left(  0,r_{0}\right)  \cap\left\{  \bigcup
_{j=1}^{N_{n}}C\left(  h_{j},r_{1}\right)  \right\}  \text{.}%
\]
Therefore,
\[
\bigcup_{j=1}^{N_{n}}S_{j}=C\left(  0,r_{0}\right)  \supset\left[
-c_{n},c_{n}\right]  ^{m},
\]
indicating (\ref{cell-4}). Similarly, we have
\begin{align*}
\bigcup_{j=1}^{\infty}S_{j}  &  =C\left(  0,r_{0}\right)  \cup\bigcup
_{j=N_{n}+1}^{\infty}S_{j}=C\left(  0,r_{0}\right)  \cup\bigcup_{j=N_{n}%
+1}^{\infty}C\left(  h_{j},r_{1}\right) \\
&  \supset\left[  -c_{n},c_{n}\right]  ^{m}\cup\bigcup_{j=N_{n}+1}^{\infty
}C\left(  h_{j},\frac{1}{2}\right)  =%
\mathbb{R}
^{m},
\end{align*}
which is (\ref{cell-5}).

Next, we compose $S_{\left(  i_{1},\cdots,i_{k}\right)  }$ . For each $k\geq
2$, let $h_{i_{1},\cdots,i_{k}}$ ($i_{k}=1$,$\cdots$, $5^{m}$) be an element
of $2^{-k+1}%
\mathbb{Z}
^{m}$ with $h_{i_{1},\cdots,i_{k}}\in C\left(  h_{i_{1},\cdots,i_{k-1}%
},2^{-k+2}\right)  $. Then, we define, recursively,
\begin{align*}
S_{\left(  i_{1},\cdots,i_{k-1},1\right)  }  &  :=S_{\left(  i_{1}%
,\cdots,i_{k-1}\right)  }\cap C\left(  h_{i_{1},\cdots,i_{k-1},1}%
,r_{k}\right)  ,\\
\,\,S_{\left(  i_{1},\cdots,i_{k}\right)  }  &  :=S_{\left(  i_{1}%
,\cdots,i_{k-1}\right)  }\cap\left\{  C\left(  h_{i_{1},\cdots,i_{k-1},i_{k}%
},r_{k}\right)  -\bigcup_{j=1}^{i_{k}-1}S_{\left(  i_{1},\cdots,i_{k-1}%
,j\right)  }\right\}  .
\end{align*}

Since
\begin{align*}
\bigcup_{j=1}^{5^{m}}S_{\left(  i_{1},\cdots,i_{k-1},j\right)  }  &
\supset\bigcup_{j=1}^{5^{m}}C\left(  h_{i_{1},\cdots,i_{k-1},j},2^{-k}\right)
\\
&  =C\left(  h_{i_{1},\cdots,i_{k-1}},2^{-k+2}+2^{-k}\right)
\end{align*}
and
\[
C\left(  h_{i_{1},\cdots,i_{k-1}},2^{-k+2}+2^{-k}\right)  \supset C\left(
h_{i_{1},\cdots,i_{k-1}},r_{k-1}\right)  \supset S_{\left(  i_{1}%
,\cdots,i_{k-1}\right)  },
\]
we have
\[
\bigcup_{j=1}^{5^{m}}S_{\left(  i_{1},\cdots,i_{k-1},j\right)  }\supset
S_{\left(  i_{1},\cdots,i_{k-1}\right)  },
\]
which implies (\ref{cell-6}).

(\ref{cell-1}) is trivial by composition. (\ref{cell-2}) is due to
\[
\partial S_{\left(  i_{1},\cdots,i_{k-1},i_{k}\right)  }\subset\partial
S_{\left(  i_{1},\cdots,i_{k-1}\right)  }\cup\bigcup_{j=1}^{i_{k}}\partial
C\left(  h_{i_{1},\cdots,i_{k-1},j},r_{k}\right)  .
\]
Hence, by (\ref{Pr=0-0}) and (\ref{Pr=0}), recursively we have (\ref{cell-3}).
(\ref{cell-2}) is obvious from that $S_{\left(  i_{1},\cdots,i_{k}\right)  }$
is a subset of $C\left(  h_{i_{1},\cdots,i_{k}},r_{k}\right)  $.
\end{proof}

\begin{lemma}
\label{lem:skorohod}Suppose $\theta\rightarrow p_{\theta}$ is continuously
differentiable in quadratic mean, and (\ref{EZ<infty}) holds. Also, let
$\Theta_{0}$ be a countable subset of $\Theta$. Then, there are random
variables $\eta_{\theta}$ and $\eta_{\theta}^{n}$ ($n\geq1$) over $\left(
\left[  0,1\right]  ,\mathcal{B}\left(  \left[  0,1\right]  \times%
\mathbb{R}
^{m}\right)  ,\nu\right)  $, such that $\nu$ is Lebesgue measure,
\begin{align}
&  \mathcal{L}\left(  \eta_{\theta}|\nu\right)  =\mathrm{N}\left(
0,J_{\theta}\right)  ,\,\,\,\,\mathcal{L}\left(  \eta_{\theta}^{n}|\nu\right)
=\mathcal{L}\left(  \ell_{\theta}^{n}|P_{\theta}^{n}\right)
,\label{skorohod-1}\\
&  \lim_{n\rightarrow\infty}\sup_{\theta\in K\cap\Theta_{0}}\nu\left\{
\left\Vert \eta_{\theta}^{n}-\eta_{\theta}\right\Vert \geq\varepsilon\right\}
=0.\label{skorohod-2}%
\end{align}

\end{lemma}

\begin{proof}
Let $S_{\left(  i_{1},\cdots,i_{k}\right)  }$ be as of Lemma\thinspace
\ref{lem:cell}, and for each $k$, and order
\[
\left\{  \left(  i_{1},\cdots,i_{k}\right)  ;i_{1}\in%
\mathbb{N}
\text{,}1\leq i_{j}\leq5^{m}\right\}
\]
lexicographically. For $\theta\in\Theta_{0}$, define intervals $\Delta
_{\theta}^{n}\left(  i_{1},\cdots,i_{k}\right)  $ of the form $[a,b)$ in
$[0,1)$ such that the length of $\Delta_{\theta}^{n}\left(  i_{1},\cdots
,i_{k}\right)  $ is $P_{\theta}^{n}\left\{  \ell_{\theta}^{n}\in S_{\left(
i_{1},\cdots,i_{n}\right)  }\right\}  $, and that, with $\left(  j_{1}%
,\cdots,j_{k}\right)  >\left(  i_{1},\cdots,i_{k}\right)  $, the left end
point of $\Delta_{\theta}^{n}\left(  j_{1},\cdots,j_{k}\right)  $ lies to the
right of $\Delta_{\theta}^{n}\left(  i_{1},\cdots,i_{k}\right)  $. Then, we
have
\[
\bigcup_{i_{1}\in%
\mathbb{N}
\text{,}1\leq i_{j}\leq5^{m}}\Delta_{\theta}^{n}\left(  i_{1},\cdots
,i_{k}\right)  =[0,1).
\]

If $P_{\theta}^{n}\left\{  \ell_{\theta}^{n}\in S_{\left(  i_{1},\cdots
,i_{n}\right)  }\right\}  $ is non-zero for some $n$, by (\ref{cell-3}), its
interior is non-empty. Thus we may take a point $x_{\left(  i_{1},\cdots
,i_{k}\right)  }$ in its interior. For $\varpi\in\lbrack0,1]$, define
\[
\eta_{\theta}^{n,k}\left(  \varpi\right)  :=x_{\left(  i_{1},\cdots
,i_{k}\right)  },\,\,\varpi\in\Delta_{\theta}^{n}\left(  i_{1},\cdots
,i_{k}\right)  .
\]
Then,
\begin{equation}
\left\Vert \eta_{\theta}^{n,k}\left(  \varpi\right)  -\eta_{\theta
}^{n,k+k^{\prime}}\left(  \varpi\right)  \right\Vert \leq\sqrt{m}2^{-k+2},
\label{eta-cauchy}%
\end{equation}
making the sequence $\left\{  \eta_{\theta}^{n,k}\left(  \varpi\right)
\right\}  _{k=1}^{\infty}$ Cauchy for each $\varpi$,$n$, and $\theta$. Hence,
$\eta_{\theta}^{n}\left(  \varpi\right)  :=$ $\lim_{k\rightarrow\infty}%
\eta_{\theta}^{n,k}\left(  \varpi\right)  $ exists.

Define the intervals $\Delta_{\theta}\left(  i_{1},\cdots,i_{k}\right)  $ of
the form $[a,b)$ in $[0,1)$ such that the length of $\Delta_{\theta}\left(
i_{1},\cdots,i_{k}\right)  $ is $P_{\mathrm{N}\left(  0,J_{\theta}\right)
}\left(  S_{\left(  i_{1},\cdots,i_{n}\right)  }\right)  $, and that, with
$\left(  j_{1},\cdots,j_{k}\right)  >\left(  i_{1},\cdots,i_{k}\right)  $, the
left end point of $\Delta_{\theta}\left(  j_{1},\cdots,j_{k}\right)  $ lies to
the right of $\Delta_{\theta}\left(  i_{1},\cdots,i_{k}\right)  $. Also, one
can define $\eta_{\theta}^{k}\left(  \varpi\right)  $ and $\eta_{\theta
}\left(  \varpi\right)  $ in the parallel manner with $\eta_{\theta}%
^{n,k}\left(  \varpi\right)  $ and $\eta_{\theta}^{n}\left(  \varpi\right)  $. \ 

\ Then, by (\ref{cell-3}) and the multi-dimensional Berry Esseen theorem
(Corollary 11.1 of \cite{Dasgupta}), we have
\[
\sup_{\theta\in K\cap\Theta_{0}}\left\vert \nu\left(  \Delta_{\theta}%
^{n}\left(  i_{1},\cdots,i_{k}\right)  \right)  -\nu\left(  \Delta_{\theta
}\left(  i_{1},\cdots,i_{k}\right)  \right)  \right\vert \leq\frac{\beta
}{\sqrt{n}},
\]
where $\beta:=400m^{1/4}\sup_{\theta\in K}\mathrm{E}_{\theta}\left\Vert
J_{\theta}^{-1}\ell_{\theta}\right\Vert ^{3}.$Therefore,
\[
\nu\left(  \Delta_{\theta}^{n}\left(  i_{1},\cdots,i_{k}\right)
\triangle\,\Delta_{\theta}\left(  i_{1},\cdots,i_{k}\right)  \right)
\leq\frac{2\beta5^{mk}i_{1}}{\sqrt{n}}.
\]
Also, by Markov's inequality,
\[
\sum_{j=N_{n}+1}^{\infty}\nu\left(  \Delta_{\theta}^{n}\left(  j\right)
\right)  \leq\frac{\sup_{\theta\in K}\mathrm{tr}\,J_{\theta}}{c_{n}^{2}%
},\,\sum_{j=N_{n}+1}^{\infty}\nu\left(  \Delta_{\theta}\left(  j\right)
\right)  \leq\frac{\sup_{\theta\in K}\mathrm{tr}\,J_{\theta}}{c_{n}^{2}}%
\]
Thus,
\begin{align*}
&  \sup_{\theta\in K\cap\Theta_{0}}\sum_{i_{1}\in%
\mathbb{N}
,1\leq i_{j}\leq5^{m}}\nu\left(  \Delta_{\theta}^{n}\left(  i_{1},\cdots
,i_{k}\right)  \triangle\,\Delta_{\theta}\left(  i_{1},\cdots,i_{k}\right)
\right) \\
&  \leq\frac{2\sup_{\theta\in K}\mathrm{tr}\,J_{\theta}}{c_{n}^{2}}%
+\frac{\beta5^{2mk}\left(  \left(  2c_{n}+1\right)  ^{m}+1\right)  ^{2}}%
{\sqrt{n}}.
\end{align*}
Here, set
\[
k=k_{n}:=\frac{\ln n}{16m\ln5},\,\,\,c_{n}:=n^{\frac{1}{16m}}.
\]
Then,%
\begin{align*}
&  \sup_{\theta\in K\cap\Theta_{0}}\sum_{i_{1}\in%
\mathbb{N}
,1\leq i_{j}\leq5^{m}}\nu\left(  \Delta_{\theta}^{n}\left(  i_{1}%
,\cdots,i_{k_{n}}\right)  \triangle\,\Delta_{\theta}\left(  i_{1}%
,\cdots,i_{k_{n}}\right)  \right) \\
&  =O\left(  n^{-\frac{1}{8m}}\right)  +O\left(  n^{-1/4}\right)
\rightarrow0,\,n\rightarrow\infty.
\end{align*}
Therefore,
\begin{align}
&  \lim_{n\rightarrow\infty}\sup_{\theta\in K\cap\Theta_{0}}\nu\left\{
\eta_{\theta}^{n,k_{n}}\left(  \varpi\right)  \neq\eta_{\theta}^{k_{n}%
}\ \left(  \varpi\right)  \right\} \nonumber\\
&  \leq\lim_{n\rightarrow\infty}\sup_{\theta\in K}\sum_{i_{1}\in%
\mathbb{N}
,1\leq i_{j}\leq5^{m}}\nu\left(  \Delta_{\theta}^{n}\left(  i_{1}%
,\cdots,i_{k_{n}}\right)  \triangle\,\Delta_{\theta}\left(  i_{1}%
,\cdots,i_{k_{n}}\right)  \right)  =0. \label{eta-1}%
\end{align}

Observe, due to (\ref{eta-cauchy})
\begin{align*}
\left\Vert \eta_{\theta}^{n}\left(  \varpi\right)  -\eta_{\theta}^{n,k_{n}%
}\left(  \varpi\right)  \right\Vert  &  =\lim_{k^{\prime}\rightarrow\infty
}\left\Vert \eta_{\theta}^{n,k^{\prime}}\left(  \varpi\right)  -\eta_{\theta
}^{n,k_{n}}\left(  \varpi\right)  \right\Vert \leq\sqrt{m}2^{-k_{n}+2},\\
\left\Vert \eta_{\theta}\left(  \varpi\right)  -\eta_{\theta}^{k_{n}}\left(
\varpi\right)  \right\Vert  &  =\lim_{k^{\prime}\rightarrow\infty}\left\Vert
\eta_{\theta}^{k^{\prime}}\left(  \varpi\right)  -\eta_{\theta}^{k_{n}}\left(
\varpi\right)  \right\Vert \leq\sqrt{m}2^{-k_{n}+2}.
\end{align*}
Therefore, due to
\begin{align*}
&  \left\Vert \eta_{\theta}^{n}\left(  \varpi\right)  -\eta_{\theta}\left(
\varpi\right)  \right\Vert \\
&  \leq\left\Vert \eta_{\theta}^{n}\left(  \varpi\right)  -\eta_{\theta
}^{n,k_{n}}\left(  \varpi\right)  \right\Vert +\left\Vert \eta_{\theta
}^{n,k_{n}}\left(  \varpi\right)  -\eta_{\theta}^{k_{n}}\left(  \varpi\right)
\right\Vert +\left\Vert \eta_{\theta}^{k_{n}}\left(  \varpi\right)
-\eta_{\theta}\left(  \varpi\right)  \right\Vert ,
\end{align*}
and (\ref{eta-1}), we have
\begin{align*}
&  \sup_{\theta\in K\cap\Theta_{0}}\nu\left\{  \left\Vert \eta_{\theta}%
^{n}-\eta_{\theta}\right\Vert \geq\varepsilon\right\} \\
&  \leq\sup_{\theta\in K\cap\Theta_{0}}\nu\left\{  \left\Vert \eta_{\theta
}^{n,k_{n}}-\eta_{\theta}^{k_{n}}\right\Vert +2\sqrt{m}2^{-k_{n}+2}%
\geq\varepsilon\right\} \\
&  \rightarrow0,\,n\rightarrow\infty,
\end{align*}
which is (\ref{skorohod-2}).

To prove (\ref{skorohod-1}), observe that every open set in $%
\mathbb{R}
^{m}$ can be expressed as a disjoint countable union of $S_{\left(
i_{1},\cdots,i_{k}\right)  }$'s. Therefore, for any open set $G$, by Fatou's
lemma,
\[
\varliminf_{k\rightarrow\infty}\nu\left\{  \eta_{\theta}^{n,k}\in G\right\}
\geq P_{\theta}^{n}\left\{  \ell_{\theta}^{n}\in G\right\}  .
\]
Hence, by Portmanteau theorem (Lemma 2.2 of \cite{vandervaart}),
$\,\lim_{k\rightarrow\infty}\mathcal{L}\left(  \eta_{\theta}^{n,k}|\nu\right)
=\mathcal{L}\left(  \ell_{\theta}^{n}|P_{\theta}^{n}\right)  $. Since
\ $\lim_{k\rightarrow\infty}\eta_{\theta}^{n,k}=\eta_{\theta}^{n}$ almost
surely, we have the second identity of (\ref{skorohod-1}). The first identity
is proved parallelly.
\end{proof}

\begin{lemma}
\label{lem:reconstruction}Suppose  $\theta\rightarrow p_{\theta}$ is
continuously differentiable in quadratic mean, and (\ref{EZ<infty}) holds.
Also, let $\Theta_{0}$ be a countable subset of $\Theta$. Then, there are
probability measure $\tilde{P}_{\theta}^{n}$ over a measurable space $\left(
\Omega^{n}\times\Omega^{\prime},\mathcal{X}^{n}\otimes\mathcal{X}^{\prime
}\right)  $, where $\left(  \Omega^{\prime},\mathcal{X}^{\prime}\right)
:=\left(
\mathbb{R}
^{m}\times\left[  0,1\right]  ,\mathcal{B}\left(
\mathbb{R}
^{m}\times\left[  0,1\right]  \right)  \right)  $, $n\geq1$, and random
variables $\lambda^{\prime n}$, $n\geq1$ over $\left(  \Omega^{n}\times
\Omega^{\prime},\mathcal{X}^{n}\otimes\mathcal{X}^{\prime},\tilde{P}_{\theta
}^{n}\right)  $, such that, $\tilde{P}_{\theta}^{n}$ is an extension of
$P_{\theta}^{n}$ and
\begin{align}
&  \lambda^{\prime n}\sim\mathrm{N}\left(  0,J_{\theta}\right)
,\label{reconstruction-1}\\
&  \lim_{n\rightarrow\infty}\sup_{\theta\in K\cap\Theta_{0}}\tilde{P}_{\theta
}^{n}\left\{  \left\Vert \ell_{\theta}^{n}-\lambda^{\prime n}\right\Vert
\geq\varepsilon\right\}  =0,\label{reconstruction-2}%
\end{align}
for any compact set $K\subset\Theta$.
\end{lemma}

\begin{proof}
The proof much draws upon the second proof of Lemma 2.2 of \cite{Shiryaev}.
Define a kernel $K_{\theta}^{n}\left(  x,\mathrm{d}y\right)  $ from $\left(
\mathbb{R}
^{m},\mathcal{B}\left(
\mathbb{R}
^{m}\right)  \right)  $ to $\left(  \left[  0,1\right]  ,\mathcal{B}\left(
\left[  0,1\right]  \right)  \right)  $ by the identity
\begin{equation}
\delta_{\eta_{\theta}^{n}\left(  y\right)  }\left(  \mathrm{d}x\right)
\nu\left(  \mathrm{d}y\right)  =R_{\theta}^{n}\left(  \mathrm{d}x\right)
K_{\theta}^{n}\left(  x,\mathrm{d}y\right)  , \label{measure}%
\end{equation}
where $\delta_{y}$ is Dirac measure, $\nu$ is the Lebesgue measure,
$\eta_{\theta}^{n}$ is as of Lemma\thinspace\ref{lem:skorohod}, and
$R_{\theta}^{n}=\mathcal{L}\left(  \ell_{\theta}^{n}|P_{\theta}^{n}\right)
=\mathcal{L}\left(  \eta_{\theta}^{n}|\nu\right)  $. (Since $\left[
0,1\right]  $ is Polish, such $K_{\theta}^{n}$ exists, see 342E of
\cite{Suugaku}. ) Define, with $\tilde{\omega}^{n}=\left(  \omega
^{n},x,y\right)  \in\Omega^{n}\times\Omega^{\prime}$,
\begin{align*}
\tilde{P}_{\theta}^{n}\left(  \mathrm{d}\tilde{\omega}^{n}\right)   &
:=P_{\theta}^{n}\left(  \mathrm{d}\omega^{n}\right)  \delta_{\ell_{\theta}%
^{n}\left(  \omega^{n}\right)  }\left(  \mathrm{d}x\right)  K_{\theta}%
^{n}\left(  x,\mathrm{d}y\right)  ,\\
\lambda^{\prime n}\left(  \tilde{\omega}^{n}\right)   &  :=\eta_{\theta
}\left(  y\right)  ,
\end{align*}
where $\eta_{\theta}$ is as of Lemma\thinspace\ref{lem:skorohod}.

Since the restriction of $\tilde{P}_{\theta}^{n}$ on $\left(  \left[
0,1\right]  ,\mathcal{B}\left(  \left[  0,1\right]  \right)  \right)  $ is
$\nu$,
\[
\mathcal{L}\left(  \lambda^{\prime n}\left(  \tilde{\omega}^{n}\right)
|\tilde{P}_{\theta}^{n}\right)  =\mathcal{L}\left(  \eta_{\theta}\left(
y\right)  |\tilde{P}_{\theta}^{n}\right)  =\mathcal{L}\left(  \eta_{\theta
}\left(  y\right)  |\nu\right)  =\mathrm{N}\left(  0,J_{\theta}\right)  .
\]
Hence, (\ref{reconstruction-1}) is shown.

By abusing the notation, we denote the extension of $\ell_{\theta}^{n}%
:\Omega^{n}\rightarrow%
\mathbb{R}
^{m}$ to $\Omega^{n}\times\Omega^{\prime}\rightarrow%
\mathbb{R}
^{m}$ also by $\ell_{\theta}^{n}$: in other words,
\[
\ell_{\theta}^{n}\left(  \omega^{n},x,y\right)  :=\ell_{\theta}^{n}\left(
\omega^{n}\right)  .
\]
To verify (\ref{reconstruction-2}), we show%
\begin{equation}
\ell_{\theta}^{n}\left(  \tilde{\omega}^{n}\right)  =\eta_{\theta}^{n}\left(
y\right)  ,\,\,\,\tilde{P}_{\theta}^{n}\text{-a.s.}. \label{l=eta}%
\end{equation}
Observe that restriction of $\tilde{P}_{\theta}^{n}$ to $\left(
\Omega^{\prime},\mathcal{X}^{\prime}\right)  =\left(
\mathbb{R}
^{m}\times\left[  0,1\right]  ,\mathcal{B}\left(
\mathbb{R}
^{m}\times\left[  0,1\right]  \right)  \right)  $ is (\ref{measure}).
Therefore, we have%
\begin{align*}
\tilde{P}_{\theta}^{n}\left(  \left\{  \ell_{\theta}^{n}\left(  \tilde{\omega
}^{n}\right)  =x\right\}  \right)   &  =\int_{\Omega^{n}}\int_{%
\mathbb{R}
^{m}}I_{\left\{  \ell_{\theta}^{n}\left(  \tilde{\omega}^{n}\right)
=x\right\}  }P_{\theta}^{n}\left(  \mathrm{d}\omega^{n}\right)  \delta
_{\ell_{\theta}^{n}\left(  \tilde{\omega}^{n}\right)  }\left(  \mathrm{d}%
x\right) \\
&  =\int_{\Omega^{n}}P_{\theta}^{n}\left(  \mathrm{d}\omega^{n}\right)  \int_{%
\mathbb{R}
^{m}}I_{\left\{  \ell_{\theta}^{n}\left(  \tilde{\omega}^{n}\right)
=x\right\}  }\delta_{\ell_{\theta}^{n}\left(  \tilde{\omega}^{n}\right)
}\left(  \mathrm{d}x\right) \\
&  =1,
\end{align*}
and%
\begin{align*}
\tilde{P}_{\theta}^{n}\left(  \left\{  \eta_{\theta}^{n}\left(  y\right)
=x\right\}  \right)   &  =\int_{%
\mathbb{R}
^{m}}\int_{\left[  0,1\right]  }I_{\left\{  \eta_{\theta}^{n}\left(  y\right)
=x\right\}  }\delta_{\eta_{\theta}^{n}\left(  y\right)  }\left(
\mathrm{d}x\right)  \nu\left(  \mathrm{d}y\right) \\
&  =\int_{\left[  0,1\right]  }\nu\left(  \mathrm{d}y\right)  \int_{\left[
0,1\right]  }I_{\left\{  \eta_{\theta}^{n}\left(  y\right)  =x\right\}
}\delta_{\eta_{\theta}^{n}\left(  y\right)  }\left(  \mathrm{d}x\right) \\
&  =1.
\end{align*}
Thus, (\ref{l=eta}) is shown.

By (\ref{l=eta}) and the definition of $\lambda^{\prime n}$,
\begin{align*}
\sup_{\theta\in K\cap\Theta_{0}}\tilde{P}_{\theta}^{n}\left\{  \left\Vert
\ell_{\theta}^{n}-\lambda^{\prime n\text{,}}\right\Vert \geq\varepsilon
\right\}   &  =\sup_{\theta\in K\cap\Theta_{0}}\tilde{P}_{\theta}^{n}\left\{
\left\Vert \eta_{\theta}^{n}-\eta_{\theta}\right\Vert \geq\varepsilon\right\}
\\
&  =\sup_{\theta\in K\cap\Theta_{0}}\nu\left\{  \left\Vert \eta_{\theta}%
^{n}-\eta_{\theta}\right\Vert \geq\varepsilon\right\} \\
&  \rightarrow0,\,\,n\rightarrow\infty.
\end{align*}

\end{proof}

\begin{lemma}
\label{lem:reconstruction-1}Suppose $\theta\rightarrow p_{\theta}$ is
differentiable in quadratic mean, and (\ref{EZ<infty}) holds. Then, there are
probability measure $\tilde{P}_{\theta}^{n}$ over a measurable space $\left(
\Omega^{n}\times\Omega^{\prime},\mathcal{X}^{n}\otimes\mathcal{X}^{\prime
}\right)  $, where $\left(  \Omega^{\prime},\mathcal{X}^{\prime}\right)
:=\left(
\mathbb{R}
^{m}\times\left[  0,1\right]  ,\mathcal{B}\left(
\mathbb{R}
^{m}\times\left[  0,1\right]  \right)  \right)  $, $n\geq1$, and random
variables $\lambda^{\prime n}$, $n\geq1$ over $\left(  \Omega^{n}\times
\Omega^{\prime},\mathcal{X}^{n}\otimes\mathcal{X}^{\prime},\tilde{P}_{\theta
}^{n}\right)  $, such that, $\tilde{P}_{\theta}^{n}$ is an extension of
$P_{\theta}^{n}$ and
\begin{align}
&  \lambda^{\prime n}\sim\mathrm{N}\left(  0,J_{\theta}\right)  ,\\
&  \lim_{n\rightarrow\infty}\tilde{P}_{\theta}^{n}\left\{  \left\Vert
\ell_{\theta}^{n}-\lambda^{\prime n}\right\Vert \geq\varepsilon\right\}  =0,
\end{align}
for any compact set $K\subset\Theta$.
\end{lemma}

\begin{proof}
This is only the combination of Lemma 2.2 of \cite{Shiryaev} and the central
limit theorem.
\end{proof}

\begin{theorem}
\label{th:uniform-lan}Suppose  $\theta\rightarrow p_{\theta}$ is continuously
differentiable in quadratic mean, and (\ref{EZ<infty}) holds. Also, let
$\Theta_{0}$ be a countable subset of $\Theta$. Then, there are probability
measures $\tilde{P}_{\theta}^{n}$ over measurable spaces $\left(  \Omega
^{n}\times\Omega^{\prime},\mathcal{X}^{n}\otimes\mathcal{X}^{\prime}\right)
$, where $\left(  \Omega^{\prime},\mathcal{X}^{\prime}\right)  :=\left(
\mathbb{R}
^{m}\times\left[  0,1\right]  ,\mathcal{B}\left(
\mathbb{R}
^{m}\times\left[  0,1\right]  \right)  \right)  $, $n\geq1$, and random
variables $\lambda_{h}^{n}$, $n\geq1$ over $\left(  \Omega^{n}\times
\Omega^{\prime},\mathcal{X}^{n}\otimes\mathcal{X}^{\prime},\tilde{P}_{\theta
}^{n}\right)  $, such that, $\tilde{P}_{\theta}^{n}$ is an extension of
$P_{\theta}^{n}$ and
\begin{align*}
&  \lim_{n\rightarrow\infty}\sup_{h\in K^{\prime}}\sup_{\theta\in K\cap
\Theta_{0}}\left\Vert \tilde{P}_{\theta+hn^{-1/2}}^{n}-Q_{\theta,h}%
^{n}\right\Vert _{1}=0,\\
&  \mathcal{L}\left(  \lambda_{h}^{n}\right)  =\mathrm{N}\left(  h,J_{\theta
}^{-1}\right)  ,\\
&  Q_{\theta,h}^{n}\left(  A\right)  :=\mathrm{E}^{\lambda^{n}}R_{\theta}%
^{n}\left(  A|\lambda_{h}^{n}\right)  .
\end{align*}
Here, $K$ is an arbitrary compact set in $\Theta$, $K^{\prime}$ is an
arbitrary compact set in $%
\mathbb{R}
^{m}$, and $R_{\theta}^{n}\left(  \cdot|\lambda^{n}\right)  $ is a measure on
$\left(  \Omega^{n}\times\Omega^{\prime},\mathcal{X}^{n}\otimes\mathcal{X}%
^{\prime}\right)  $, which may depend on $\theta$, but is independent of $h$.
\end{theorem}

\begin{proof}
We use Lemma\thinspace\ref{lem:converge}, with $f\left(  x\right)  =e^{x}$,
$t=\left(  \theta,h\right)  $, $X_{n,t}:=h^{T}\lambda^{\prime n}-\frac{1}%
{2}h^{T}J_{\theta}h$ and $Y_{n,t}:=h^{T}\ell_{\theta}^{n}-\frac{1}{2}%
h^{T}J_{\theta}h$. \ Obviously, (\ref{prob-converge}) and (\ref{f-1}) are satisfied.

Due to (\ref{E[e>a]}), we have (\ref{f-2}):
\[
\lim_{n\rightarrow\infty}\sup_{h\in K^{\prime}}\sup_{\theta\in K\cap\Theta
_{0}}\mathrm{E}_{\theta}^{n}\left[  Z_{\theta,h}\left(  \ell_{\theta}%
^{n}\right)  :Z_{\theta,h}\left(  \ell_{\theta}^{n}\right)  >a\right]
\rightarrow0,a\rightarrow\infty.
\]
Due to (\ref{reconstruction-1}) , we have%
\begin{align*}
&  \lim_{n\rightarrow\infty}\sup_{h\in K^{\prime}}\sup_{\theta\in K\cap
\Theta_{0}}\mathrm{\tilde{E}}_{\theta}^{n}\left[  Z_{\theta,h}\left(
\lambda^{\prime n}\right)  :Z_{\theta,h}\left(  \lambda^{\prime n}\right)
>a\right] \\
&  =\sup_{h\in K^{\prime}}\sup_{\theta\in K\cap\Theta_{0}}\mathrm{E}%
^{X}\left[  Z_{\theta,h}\left(  X\right)  :Z_{\theta,h}\left(  X\right)
>a\right]  \rightarrow0,\,\,a\rightarrow\infty,
\end{align*}
where $\mathcal{L}\left(  X\right)  =\mathrm{N}\left(  0,J_{\theta}%
^{-1}\right)  $. Thus (\ref{f-3}). Therefore, we have
\[
\lim_{n\rightarrow\infty}\sup_{h\in K^{\prime}}\sup_{\theta\in K\cap\Theta
_{0}}\mathrm{\tilde{E}}_{\theta}^{n}\left\vert Z_{\theta,h}\left(
\ell_{\theta}^{n}\right)  -Z_{\theta,h}\left(  \lambda^{\prime n}\right)
\right\vert =0.
\]
Therefore, combining Lemma\thinspace\ref{lem:EZ-EZl},
\begin{align*}
&  \sup_{h\in K^{\prime}}\sup_{\theta\in K}\mathrm{\tilde{E}}_{\theta}%
^{n}\left\vert Z_{\theta,h}^{n}-Z_{\theta,h}\left(  \lambda^{\prime n}\right)
\right\vert \\
&  \leq\sup_{h\in K^{\prime}}\sup_{\theta\in K}\mathrm{E}_{\theta}%
^{n}\left\vert Z_{\theta,h}^{n}-Z_{\theta,h}\left(  \ell_{\theta}^{n}\right)
\right\vert +\sup_{h\in K^{\prime}}\sup_{\theta\in K}\mathrm{\tilde{E}%
}_{\theta}^{n}\left\vert Z_{\theta,h}\left(  \ell_{\theta}^{n}\right)
-Z_{\theta,h}\left(  \lambda^{\prime n}\right)  \right\vert \\
&  \rightarrow0.
\end{align*}

Let $\tilde{W}_{\theta}^{n}$ be the random variable with $\mathcal{L}\left(
\tilde{W}_{\theta}^{n}\right)  =\tilde{P}_{\theta}^{n}$. Then
\begin{align*}
&  \mathrm{\tilde{E}}_{\theta}^{n}Z_{\theta,h}\left(  \lambda^{\prime
n}\right)  I_{A}\left(  \tilde{W}_{\theta}^{n}\right) \\
&  =\int\mathrm{\tilde{E}}_{\theta}^{n}\left[  Z_{\theta,h}\left(
\lambda^{\prime n}\right)  I_{A}\left(  \tilde{W}_{\theta}^{n}\right)
|\lambda^{\prime n}=x\right]  \frac{e^{-\frac{1}{2}x^{T}J_{\theta}^{-1}%
x}\mathrm{d}x}{\left(  2\pi\right)  ^{m/2}\left(  \det J_{\theta}\right)
^{1/2}}\\
&  =\int\mathrm{\tilde{E}}_{\theta}^{n}\left[  I_{A}\left(  \tilde{W}_{\theta
}^{n}\right)  |\lambda^{\prime n}=x\right]  \exp\left\{  ^{-\frac{1}{2}%
x^{T}J_{\theta}^{-1}x}+h^{T}x-\frac{1}{2}h^{T}J_{\theta}h\right\}
\frac{\mathrm{d}x}{\left(  2\pi\right)  ^{m/2}\left(  \det J_{\theta}\right)
^{1/2}}\\
&  =\int\mathrm{\tilde{E}}_{\theta}^{n}\left[  I_{A}\left(  \tilde{W}_{\theta
}^{n}\right)  |\lambda^{\prime n}=J_{\theta}x\right]  \exp\left\{  -\frac
{1}{2}\left(  x-h\right)  ^{T}J_{\theta}\left(  x-h\right)  \right\}
\frac{\left(  \det J_{\theta}\right)  ^{1/2}}{\left(  2\pi\right)  ^{m/2}%
}\mathrm{d}x.
\end{align*}
Since $\Omega^{n}\times\Omega^{\prime}$ is Polish, there is a nice version of
$R_{\theta}^{n}\left(  A|x\right)  :=\mathrm{\tilde{E}}_{\theta}^{n}\left[
I_{A}\left(  \tilde{W}_{\theta}^{n}\right)  |\lambda^{\prime n}=J_{\theta
}x\right]  $ which is a probability measure in $\mathcal{X}^{n}\times
\mathcal{X}^{\prime}$ for every $\not x  \in%
\mathbb{R}
^{m}$ (see, for example, 342E of \cite{Suugaku}). By definition,
\[
\mathrm{\tilde{E}}_{\theta}Z_{\theta,h}\left(  \lambda^{\prime n}\right)
I_{A}\left(  \tilde{W}_{\theta}^{n}\right)  =\mathrm{E}^{\lambda_{h}^{n}%
}R_{\theta}^{n}\left(  A|\lambda_{h}^{n}\right)  .
\]
Therefore, we have the assertion:
\begin{align*}
&  \sup_{h\in K^{\prime}}\sup_{\theta\in K}\left\vert \tilde{P}_{\theta
+n^{-1/2}h}^{n}\left(  A\right)  -Q_{\theta,h}^{n}\left(  A\right)
\right\vert \\
&  =\sup_{h\in K^{\prime}}\sup_{\theta\in K}\left\vert \tilde{P}%
_{\theta+n^{-1/2}h}^{n}\left(  A\right)  -\mathrm{E}^{\lambda^{n}}R_{\theta
}^{n}\left(  A|\lambda^{n}\right)  \right\vert \\
&  =\sup_{h\in K^{\prime}}\sup_{\theta\in K}\left\vert \mathrm{\tilde{E}%
}_{\theta}^{n}Z_{\theta,h}^{n}I_{A}\left(  \tilde{W}_{\theta}^{n}\right)
-\mathrm{\tilde{E}}_{\theta}^{n}Z_{\theta,h}\left(  \lambda^{\prime n}\right)
I_{A}\left(  \tilde{W}_{\theta}^{n}\right)  \right\vert \\
&  \leq\sup_{h\in K^{\prime}}\sup_{\theta\in K}\mathrm{\tilde{E}}_{\theta}%
^{n}\left\vert Z_{\theta,h}^{n}-Z_{\theta,h}\left(  \lambda^{\prime n}\right)
\right\vert \rightarrow0,\,\,n\rightarrow\infty.
\end{align*}

\end{proof}

\begin{theorem}
\label{th:lan}Suppose  $\theta\rightarrow p_{\theta}$ is differentiable in
quadratic mean, and (\ref{EZ<infty}) holds. Then, there are probability
measures $\tilde{P}_{\theta}^{n}$ over a measurable spaces $\left(  \Omega
^{n}\times\Omega^{\prime},\mathcal{X}^{n}\otimes\mathcal{X}^{\prime}\right)
$, where $\left(  \Omega^{\prime},\mathcal{X}^{\prime}\right)  :=\left(
\mathbb{R}
^{m}\times\left[  0,1\right]  ,\mathcal{B}\left(
\mathbb{R}
^{m}\times\left[  0,1\right]  \right)  \right)  $, $n\geq1$, and random
variables $\lambda_{h}^{n}$, $n\geq1$ over $\left(  \Omega^{n}\times
\Omega^{\prime},\mathcal{X}^{n}\otimes\mathcal{X}^{\prime},\tilde{P}_{\theta
}^{n}\right)  $, such that, $\tilde{P}_{\theta}^{n}$ is an extension of
$P_{\theta}^{n}$ and
\begin{align*}
&  \lim_{n\rightarrow\infty}\sup_{h\in K^{\prime}}\left\Vert \tilde{P}%
_{\theta+hn^{-1/2}}^{n}-Q_{\theta,h}^{n}\right\Vert _{1}=0,\\
&  \lambda_{h}^{n}\sim\mathrm{N}\left(  h,J_{\theta}^{-1}\right)  ,\\
&  Q_{\theta,h}^{n}\left(  A\right)  :=\mathrm{E}^{\lambda^{n}}R_{\theta}%
^{n}\left(  A|\lambda_{h}^{n}\right)  .
\end{align*}
Here, $K^{\prime}$ is an arbitrary compact set in $%
\mathbb{R}
^{m}$, and $R_{\theta}^{n}\left(  \cdot|\lambda^{n}\right)  $ is a measure on
$\left(  \Omega^{n}\times\Omega^{\prime},\mathcal{X}^{n}\otimes\mathcal{X}%
^{\prime}\right)  $, which may depend on $\theta$, but is independent of $h$.
\end{theorem}

\begin{proof}
The proof is parallel with the one of Theorem\thinspace\ref{th:uniform-lan},
except that Lemma\thinspace\ref{lem:reconstruction-1} is used instead of
Lemma\thinspace\ref{lem:reconstruction}, and that $\sup_{\theta\in\Theta}$ at
each step is removed.
\end{proof}

\subsection{Asymptotic cloner using the optimal amplifier for the Gaussian
shift family}

Hereafter, we assume the existence of a sequence $\left\{  \hat{\theta}%
^{n}\right\}  $ of estimate of $\theta$, such that
\begin{equation}
\lim_{a\rightarrow\infty}\lim_{n\rightarrow\infty}P_{\theta}^{n}\left\{
\sqrt{n}\left\Vert \hat{\theta}^{n}-\theta\right\Vert \geq a\right\}  =0.
\label{consistent}%
\end{equation}
Without loss of generality, one can suppose that
\begin{equation}
\hat{\theta}^{n}\in n^{-1/2}%
\mathbb{Z}
. \label{discrete}%
\end{equation}
If (\ref{discrete}) is not satisfied, we redefine $\hat{\theta}^{n}$ as the
closest element of $n^{-1/2}%
\mathbb{Z}
$ to $\hat{\theta}^{n}$. Obviously, newly defined $\hat{\theta}^{n}$ satisfies
(\ref{consistent}). Therefore, letting
\[
\Theta_{0}:=\left\{  k^{-1/2}\cdot l;k\in%
\mathbb{N}
,l\in%
\mathbb{Z}
\right\}  ,
\]
we can suppose
\[
\hat{\theta}^{n}\in\Theta_{0}%
\]
and the cardinality of $\Theta_{0}$ is countable.

We consider the following procedure of $\left(  n,rn\right)  $-cloner
$\Lambda_{\delta,\varepsilon}^{n,r}$. For the composition, we use the optimal
$r$-amplifier $\Lambda_{\mathrm{amp}}^{r}=\Psi_{\sqrt{r}}$ of the Gaussian
shift family $\left\{  \mathrm{N}\left(  h,J_{\theta}^{-1}\right)  \right\}
_{h\in%
\mathbb{R}
^{m}}$. Also, define
\[
L_{\theta}^{n,\varepsilon}\left(  \omega^{n}\right)  :=J_{\theta}^{-1}%
\ell_{\theta}^{n}\left(  \omega^{n}\right)  +Y_{\varepsilon},
\]
where $\mathcal{L}\left(  Y_{\varepsilon}\right)  =$ $\mathrm{N}\left(
0,\varepsilon\boldsymbol{1}\right)  $.

Then, for a given $\varepsilon>0$ and $0<\delta<1$, we construct a cloner
$\Lambda_{\delta,\varepsilon}^{n,r}$ as follows.

\begin{description}
\item[(I)] Estimate $\theta$ using $n_{1}$-data, ($n_{1}:=\delta n$) and let
$n_{2}:=\left(  1-\delta\right)  n\,.$

\item[(II)] Apply $\Lambda_{\mathrm{amp}}^{\sqrt{r/1-\delta}}$ to
$\mathcal{L}\left(  L_{\hat{\theta}^{n_{1}}}^{n_{2},\varepsilon}|P_{\theta
}^{n_{2}}\right)  $. Denote the resulting random variable by $\tilde{X}%
_{\hat{\theta}^{n_{1}}}^{n}$.$\frac{\frac{\frac{{}}{{}}}{{}}}{{}}$

\item[(III)] Generate $\left(  \omega^{rn},\omega^{\prime}\right)  $ according
to $R_{\hat{\theta}^{n_{1}}}^{rn}\left(  \cdot|\tilde{X}_{\hat{\theta}^{n_{1}%
}}^{n}\right)  $, and discard $\omega^{\prime}$.
\end{description}

The output probability distribution is
\[
\Lambda_{\delta,\varepsilon}^{n,r}\left(  P_{\theta}^{n}\right)
=\mathrm{E}^{\hat{\theta}^{n_{1}}}\mathrm{E}^{\tilde{X}_{\hat{\theta}^{n_{1}}%
}^{n}}R_{\hat{\theta}^{n_{1}}}^{rn}\left(  A\times\Omega^{\prime}|\tilde
{X}_{\hat{\theta}^{n_{1}}}^{n}\right)  .
\]
We will show this is asymptotically optimal.

\begin{lemma}
\label{lem:local-limit}Suppose  $\theta\rightarrow p_{\theta}$ is continuously
differentiable in quadratic mean, and (\ref{EZ<infty}) holds. Moreover,
suppose (\ref{J-cont}) is satisfied. Then, for any compact set $K^{\prime
}\subset%
\mathbb{R}
^{m}$,
\[
\lim_{\varepsilon\downarrow0}\lim_{n\rightarrow\infty}\sup_{h\in K^{\prime}%
}\left\Vert \mathcal{L}\left(  L_{\theta_{n,h}}^{n,\varepsilon}|P_{\theta
_{n,h}}^{n}\right)  -\mathrm{N}\left(  0,J_{\theta}^{-1}\right)  \right\Vert
_{1}=0,
\]
where $\theta_{n,h}:=\theta+n^{-1/2}h$.
\end{lemma}

\begin{proof}
\bigskip Define $h_{n}$ so that
\[
\left\Vert \mathcal{L}\left(  L_{\theta_{n,h_{n}}}^{n,\varepsilon}%
|P_{\theta_{n,h_{n}}}^{n}\right)  -\mathrm{N}\left(  0,J_{\theta}^{-1}\right)
\right\Vert _{1}\geq\sup_{h\in K^{\prime}}\left\Vert \mathcal{L}\left(
L_{\theta_{n,h}}^{n,\varepsilon}|P_{\theta_{n,h}}^{n}\right)  -\mathrm{N}%
\left(  0,J_{\theta}^{-1}\right)  \right\Vert _{1}-\varepsilon^{\prime}%
\]
holds, and let $\theta^{n}:=\theta_{n,h_{n}}$. Then, $\lim_{n\rightarrow
\infty}\theta^{n}=\theta$.

Denote by $\phi_{\theta^{\prime}}$ the characteristic function of the
distribution of $J_{\theta^{\prime}}^{-1}\ell_{\theta^{\prime}}\left(
W_{\theta,\kappa}\right)  $. Then, the density of $\mathcal{L}\left(
L_{\theta^{n}}^{n,\varepsilon}|P_{\theta^{n}}^{n}\right)  $ with respect to
Lebesgue measure is
\[
\frac{1}{2\pi}\int\left\{  \phi_{\theta^{n}}\left(  \frac{t}{\sqrt{n}}\right)
\right\}  ^{n}e^{-\frac{1}{2}\varepsilon\left\Vert t\right\Vert ^{2}}%
e^{-\sqrt{-1}t\cdot x}\mathrm{d}t.
\]
Observe
\[
\int\left\vert \left\{  \phi_{\theta^{n}}\left(  \frac{t}{\sqrt{n}}\right)
\right\}  ^{n}e^{-\frac{1}{2}\varepsilon\left\Vert t\right\Vert ^{2}}%
e^{-\sqrt{-1}t\cdot x}\right\vert \mathrm{d}t\leq\int e^{-\varepsilon
\left\Vert t\right\Vert ^{2}}\mathrm{d}t<\infty.
\]
Hence, by Lebesgue's dominated convergence theorem, we have, with
$\ f_{\mathrm{rem}}$ being as of (\ref{exp-rem}),
\begin{align*}
&  \lim_{n\rightarrow\infty}\frac{1}{2\pi}\int\left\{  \phi_{\theta^{n}%
}\left(  \frac{t}{\sqrt{n}}\right)  \right\}  ^{n}e^{-\frac{1}{2}%
\varepsilon\left\Vert t\right\Vert ^{2}}e^{-\sqrt{-1}t\cdot x}\mathrm{d}t\\
&  =\frac{1}{2\pi}\int\lim_{n\rightarrow\infty}\left\{  \phi_{\theta^{n}%
}\left(  \frac{t}{\sqrt{n}}\right)  \right\}  ^{n}e^{-\frac{1}{2}%
\varepsilon\left\Vert t\right\Vert ^{2}}e^{-\sqrt{-1}t\cdot x}\mathrm{d}t,\\
&  =\,\ \frac{1}{2\pi}\int\lim_{n\rightarrow\infty}\left\{  1-\frac{1}%
{2n}\left(  t^{T}J_{\theta^{n}}^{-1}t\right)  +f_{\mathrm{rem}}\left(
\theta,\sqrt{-1}t,n\right)  \right\}  ^{n}e^{-\frac{1}{2}\varepsilon\left\Vert
t\right\Vert ^{2}}e^{-\sqrt{-1}t\cdot x}\mathrm{d}t\\
&  =\,\,\ \frac{1}{2\pi}\int\exp\left\{  -\frac{1}{2}t^{T}\left(  J_{\theta
}^{-1}+\varepsilon^{2}\mathbf{1}\right)  t\right\}  e^{-\sqrt{-1}t\cdot
x}\mathrm{d}t\,\,\,\text{a.e.}%
\end{align*}
Here, in the third line, we used Lemma\thinspace
\ref{lem:likelihood-expansion-1} to show that the first order term of the
Taylor expansion ($=\mathrm{E}_{\theta}^{n}\ell_{\theta}^{n}$) vanishes. Also,
to obtain the fourth line, we used the inequality (\ref{rem-upper}).

Therefore, the density of $\mathcal{L}\left(  L_{\theta^{n}}^{n,\varepsilon
}|P_{\theta^{n}}^{n}\right)  $ converges to the one of $\mathrm{N}\left(
0,J_{\theta}^{-1}+\varepsilon\right)  $, as $n\rightarrow\infty$. Therefore,
By Schefe's lemma, we have
\[
\lim_{n\rightarrow\infty}\left\Vert \mathcal{L}\left(  L_{\theta^{n}%
}^{n,\varepsilon}|P_{\theta^{n}}^{n}\right)  -\mathrm{N}\left(  0,J_{\theta
}^{-1}+\varepsilon\right)  \right\Vert _{1}=0.
\]
Therefore,
\begin{align*}
&  \lim_{\varepsilon\downarrow0}\lim_{n\rightarrow\infty}\left\Vert
\mathcal{L}\left(  L_{\theta^{n}}^{n,\varepsilon}|P_{\theta^{n}}^{n}\right)
-\mathrm{N}\left(  0,J_{\theta}^{-1}\right)  \right\Vert _{1}\\
&  \leq\lim_{\varepsilon\downarrow0}\lim_{n\rightarrow\infty}\left\Vert
\mathcal{L}\left(  L_{\theta^{n}}^{n,\varepsilon}|P_{\theta^{n}}^{n}\right)
-\mathrm{N}\left(  0,J_{\theta}^{-1}+\varepsilon\right)  \right\Vert _{1}%
+\lim_{\varepsilon\downarrow0}\left\Vert \mathrm{N}\left(  0,J_{\theta}%
^{-1}\right)  -\mathrm{N}\left(  0,J_{\theta}^{-1}+\varepsilon\right)
\right\Vert _{1}\\
&  =0.
\end{align*}

Therefore,
\[
\lim_{\varepsilon\downarrow0}\lim_{n\rightarrow\infty}\sup_{h\in K^{\prime}%
}\left\Vert \mathcal{L}\left(  L_{\theta_{n,h}}^{n,\varepsilon}|P_{\theta
_{n,h}}^{n}\right)  -\mathrm{N}\left(  0,J_{\theta}^{-1}\right)  \right\Vert
_{1}\leq\varepsilon^{\prime}.
\]
Since $\varepsilon^{\prime}>0$ is arbitrary, we have the assertion.
\end{proof}

\begin{lemma}
\label{lem:||p-N||}Suppose  $\theta\rightarrow p_{\theta}$ is continuously
differentiable in quadratic mean, and (\ref{EZ<infty}) holds. Moreover, we
suppose (\ref{J-cont}) and (\ref{J-nondeg}) hold. Then, for any compact set
$K^{\prime}\in%
\mathbb{R}
^{m}$, we have
\[
\lim_{\varepsilon\downarrow0}\lim_{n\rightarrow\infty}\sup_{h\in K^{\prime}%
}\left\Vert \mathcal{L}\left(  L_{\theta_{n,h}}^{n,\varepsilon}|P_{\theta}%
^{n}\right)  -\mathrm{N}\left(  -h,J_{\theta}^{-1}\right)  \right\Vert _{1}=0,
\]
where $\theta_{n,h}:=\theta+n^{-1/2}h$
\end{lemma}

\begin{proof}
Observe, for any measurable function $f$ with $\sup_{x\in%
\mathbb{R}
^{m}}\left\vert f\left(  x\right)  \right\vert \leq1$,%
\[
\mathrm{E}^{Y_{\varepsilon}}\mathrm{E}_{\theta}^{n}\left[  f\left(
L_{\theta_{n,h}}^{n,\varepsilon}\right)  \right]  =\mathrm{E}^{Y_{\varepsilon
}}\mathrm{E}_{\theta_{n,h}}^{n}\left[  f\left(  L_{\theta_{n,h}}%
^{n,\varepsilon}\right)  Z_{\theta_{n,h},-h}^{n}\,\right]  ,
\]
Observe also, due to Lemma\thinspace\ref{lem:EZ-EZl}, with $K_{\theta}$ being
a compact subset of $\Theta$ containing $\theta$ and $K^{\prime}$ being an
arbitrary compact subset of $%
\mathbb{R}
^{m}$,%
\begin{align}
&  \sup_{h\in K^{\prime}}\left\vert \mathrm{E}^{Y_{\varepsilon}}%
\mathrm{E}_{\theta_{n,h}}^{n}\left[  f\left(  L_{\theta_{n,h}}^{n,\varepsilon
}\right)  \left(  Z_{\theta_{n,h},-h}^{n}\,-Z_{\theta_{n,h},-h}\left(
\ell_{\theta_{n,h}}^{n}\right)  \,\right)  \right]  \right\vert \nonumber\\
&  \leq\sup_{h\in K^{\prime}}\mathrm{E}_{\theta_{n,h}}^{n}\left[  \left\vert
Z_{\theta_{n,h},-h}^{n}\,-Z_{\theta_{n,h},-h}\left(  \ell_{\theta_{n,h}}%
^{n}\right)  \,\right\vert \,\right]  \nonumber\\
&  \leq\sup_{h\in K^{\prime}}\sup_{\theta^{\prime}\in K_{\theta}}%
\mathrm{E}_{\theta^{\prime}}^{n}\left[  \left\vert Z_{\theta^{\prime},-h}%
^{n}\,-Z_{\theta^{\prime},-h}\left(  \ell_{\theta^{\prime}}^{n}\right)
\,\right\vert \,\right]  \nonumber\\
&  \rightarrow0,\,\,n\rightarrow\infty.\label{EfZ}%
\end{align}
Therefore,\quad we have to evaluate
\begin{align*}
&  \lim_{n\rightarrow\infty}\sup_{h\in K^{\prime}}\mathrm{E}^{Y_{\varepsilon}%
}\mathrm{E}_{\theta_{n,h}}^{n}\left[  f\left(  L_{\theta_{n,h}}^{n,\varepsilon
}\right)  Z_{\theta_{n,h},-h}\,\left(  \ell_{\theta_{n,h}}^{n}\right)
\right]  \\
&  =\lim_{n\rightarrow\infty}\sup_{h\in K^{\prime}}\mathrm{E}^{Y_{\varepsilon
}}\mathrm{E}_{\theta_{n,h}}^{n}\left[  f\left(  L_{\theta_{n,h}}%
^{n,\varepsilon}\right)  Z_{\theta_{n,h},-h}\,\left(  J_{\theta_{n,h}%
}L_{\theta_{n,h}}^{n,\varepsilon}\right)  e^{h^{T}J_{\theta_{n,h}%
}Y_{\varepsilon}}\right]  \\
&  =\lim_{n\rightarrow\infty}\sup_{h\in K^{\prime}}\left\{  E_{1}+E_{2}%
+E_{3}\right\}  ,
\end{align*}
where
\begin{align*}
E_{1} &  :=\mathrm{E}^{Y_{\varepsilon}}\mathrm{E}_{\theta_{n,h}}^{n}\left[
I_{\left\{  \left\Vert L_{\theta_{n,h}}^{n,\varepsilon}\right\Vert \leq
a\,,\,\left\Vert Y_{\varepsilon}\right\Vert \leq\varepsilon^{1/4}\right\}
}f\left(  L_{\theta_{n,h}}^{n,\varepsilon}\right)  Z_{\theta_{n,h}%
,-h}\,\left(  J_{\theta_{n,h}}L_{\theta_{n,h}}^{n,\varepsilon}\right)
e^{h^{T}J_{\theta_{n,h}}Y_{\varepsilon}}\right]  ,\\
E_{2} &  :=\mathrm{E}^{Y_{\varepsilon}}\mathrm{E}_{\theta_{n,h}}^{n}\left[
I_{\left\{  \,\left\Vert Y_{\varepsilon}\right\Vert >\varepsilon
^{1/4}\right\}  }f\left(  L_{\theta_{n,h}}^{n,\varepsilon}\right)
Z_{\theta_{n,h},-h}\,\left(  J_{\theta_{n,h}}L_{\theta_{n,h}}^{n,\varepsilon
}\right)  e^{h^{T}J_{\theta_{n,h}}Y_{\varepsilon}}\right]  ,\\
E_{3} &  :=\mathrm{E}^{Y_{\varepsilon}}\mathrm{E}_{\theta_{n,h}}^{n}\left[
I_{\left\{  \left\Vert L_{\theta_{n,h}}^{n,\varepsilon}\right\Vert
>a\,,\,\left\Vert Y_{\varepsilon}\right\Vert \leq\varepsilon^{1/4}\right\}
}f\left(  L_{\theta_{n,h}}^{n,\varepsilon}\right)  Z_{\theta_{n,h}%
,-h}\,\left(  J_{\theta_{n,h}}L_{\theta_{n,h}}^{n,\varepsilon}\right)
e^{h^{T}J_{\theta_{n,h}}Y_{\varepsilon}}\right]  .
\end{align*}

The first term of the right most side of $E_{1}$ is evaluated as follows.
\begin{equation}
E_{1}=\mathrm{E}^{Y_{\varepsilon}}\mathrm{E}_{\theta_{n,h}}^{n}\left[
I_{\left\{  \left\Vert L_{\theta_{n,h}}^{n,\varepsilon}\right\Vert \leq
a\right\}  }f\left(  L_{\theta_{n,h}}^{n,\varepsilon}\right)  Z_{\theta
_{n,h},-h}\,\left(  J_{\theta_{n,h}}L_{\theta_{n,h}}^{n,\varepsilon}\right)
\mathrm{E}\left[  \left.  I_{\left\{  \left\vert Y_{\varepsilon}\right\vert
\leq\varepsilon^{1/4}\right\}  }e^{h^{T}J_{\theta_{n,h}}Y_{\varepsilon}%
}\right\vert L_{\theta_{n,h}}^{n,\varepsilon}\right]  \right]  , \label{EfZ-1}%
\end{equation}
whose second factor can be evaluated as
\begin{align}
&  \lim_{\varepsilon\downarrow0}\lim_{n\rightarrow\infty}\sup_{h\in K^{\prime
}}\left\vert \mathrm{E}\left[  \left.  I_{\left\{  \left\Vert Y_{\varepsilon
}\right\Vert \leq\varepsilon^{1/2}\right\}  }e^{h^{T}J_{\theta_{n,h}}%
Y}\right\vert L_{\theta_{n,h}}^{n,\varepsilon}\right]  -1\right\vert
\nonumber\\
&  \leq\lim_{\varepsilon\downarrow0}\lim_{n\rightarrow\infty}\sup_{h\in
K^{\prime}}e^{\left\Vert h\right\Vert \left\Vert J_{\theta_{n,h}}\right\Vert
\varepsilon^{1/4}}=0 \label{EfZ-11}%
\end{align}

To evaluate the first factor of $E_{1}$, or
\[
E_{1,1}:=\mathrm{E}^{Y_{\varepsilon}}\mathrm{E}_{\theta_{n,h}}^{n}\left[
I_{\left\{  \left\Vert L_{\theta_{n,h}}^{n,\varepsilon}\right\Vert \leq
a\right\}  }f\left(  L_{\theta_{n,h}}^{n,\varepsilon}\right)  Z_{\theta
_{n,h},-h}\,\left(  J_{\theta_{n,h}}L_{\theta_{n,h}}^{n,\varepsilon}\right)
\right]  ,
\]
observe
\begin{align*}
&  \lim_{n\rightarrow\infty}\sup_{h\in K^{\prime}}\left\vert E_{1,1}%
-\mathrm{E}^{Y_{\varepsilon}}\mathrm{E}_{\theta_{n,h}}^{n}\left[  I_{\left\{
\left\Vert L_{\theta_{n,h}}^{n,\varepsilon}\right\Vert \leq a\right\}
}f\left(  L_{\theta_{n,h}}^{n}\right)  Z_{\theta,-h}\,\left(  J_{\theta
}L_{\theta_{n,h}}^{n,\varepsilon}\right)  \right]  \right\vert \\
&  \leq\lim_{n\rightarrow\infty}\sup_{h\in K^{\prime}}\sup_{\left\Vert
L\right\Vert \leq a}\left\vert Z_{\theta_{n,h},-h}\,\left(  J_{\theta_{n,h}%
}L\right)  -Z_{\theta,-h}\,\left(  J_{\theta}L\right)  \right\vert \\
&  =\lim_{n\rightarrow\infty}\sup_{h\in K^{\prime}}\sup_{\left\Vert
L\right\Vert \leq a}\left\vert \exp\left\{  -h^{T}J_{\theta_{n,h}}L\right\}
e^{\frac{1}{2}h^{T}J_{\theta_{n,h}}h}-\exp\left\{  -h^{T}J_{\theta}L\right\}
e^{\frac{1}{2}h^{T}J_{\theta}h}\right\vert \\
&  =0.
\end{align*}
Therefore, letting $X_{h}$ be a random variable with $\mathcal{L}\left(
X_{h}\right)  =\mathrm{N}\left(  h,J_{\theta}^{-1}\right)  $,
\begin{align}
&  \lim_{\varepsilon\downarrow0}\lim_{n\rightarrow\infty}\sup_{h\in K^{\prime
}}\left\vert E_{1,1}-\mathrm{E}^{X_{-h}}\left[  f\left(  X_{-h}\right)
;\left\Vert X_{-h}\right\Vert \leq a\right]  \right\vert \nonumber\\
&  =\lim_{\varepsilon\downarrow0}\lim_{n\rightarrow\infty}\sup_{h\in
K^{\prime}}\left\vert
\begin{array}
[c]{c}%
\mathrm{E}^{Y_{\varepsilon}}\mathrm{E}_{\theta_{n,h}}^{n}\left[  f\left(
L_{\theta_{n,h}}^{n}\right)  Z_{\theta,-h}\,\left(  J_{\theta}L_{\theta_{n,h}%
}^{n,\varepsilon}\right)  :\left\Vert L_{\theta_{n,h}}^{n,\varepsilon
}\right\Vert \leq a\right] \\
-\mathrm{E}^{X_{0}}\left[  f\left(  X_{0}\right)  Z_{\theta,-h}\,\left(
J_{\theta}X_{0}\right)  :\left\Vert X_{0}\right\Vert \leq a\right]
\end{array}
\right\vert \nonumber\\
&  \leq\sup_{h\in K^{\prime}}e^{\left\Vert h\right\Vert \left\Vert J_{\theta
}\right\Vert a}e^{-\frac{1}{2}h^{T}J_{\theta}h}\lim_{\varepsilon\downarrow
0}\lim_{n\rightarrow\infty}\sup_{h\in K^{\prime}}\left\Vert \mathcal{L}\left(
L_{\theta_{n,h}}^{n,\varepsilon}|P_{\theta_{n,h}}^{n}\right)  -\mathrm{N}%
\left(  0,J_{\theta}^{-1}\right)  \right\Vert _{1}\nonumber\\
&  =0, \label{EfZ-12}%
\end{align}
where the last identity is due to Lemma\thinspace\ref{lem:local-limit}.

Therefore, by (\ref{EfZ-11}) and (\ref{EfZ-12}),
\begin{equation}
\lim_{\varepsilon\downarrow0}\lim_{n\rightarrow\infty}\sup_{h\in K^{\prime}%
}\left\vert E_{1}-\mathrm{E}\left[  f\left(  X_{-h}\right)  ;\left\Vert
X_{-h}\right\Vert \leq a\right]  \right\vert =0. \label{EfZ-EfX}%
\end{equation}

On the other hand, by (\ref{EZ<finite}), $E_{2}$the second term of the right
most side of (\ref{EfZ}) is evaluated as
\begin{align}
&  \lim_{\varepsilon\downarrow0}\lim_{n\rightarrow\infty}\sup_{h\in K^{\prime
}}E_{2}\leq\lim_{n\rightarrow\infty}\sup_{h\in K^{\prime}}\mathrm{E}%
^{Y_{\varepsilon}}\mathrm{E}_{\theta_{n,h}}^{n}\left[  Z_{\theta_{n,h}%
,-h}\,\left(  \ell_{\theta_{n,h}}^{n}\right)  :\left\Vert Y_{\varepsilon
}\right\Vert >\varepsilon^{1/4}\right] \nonumber\\
&  =\lim_{\varepsilon\downarrow0}\Pr\left\{  \left\Vert Y_{\varepsilon
}\right\Vert >\varepsilon^{1/4}\right\}  \cdot\lim_{n\rightarrow\infty}%
\sup_{h\in K^{\prime}}e^{h^{T}J_{\theta_{n,h}}h}e^{-\frac{1}{2}h^{T}%
J_{\theta_{n,h}}h}\nonumber\\
&  =0. \label{EfZ-2}%
\end{align}
Also, $E_{3}$ is evaluated as, by (\ref{EZ<finite}),
\begin{align}
&  \lim_{\varepsilon\downarrow0}\lim_{n\rightarrow\infty}\sup_{h\in K^{\prime
}}E_{3}\nonumber\\
&  \leq\lim_{\varepsilon\downarrow0}\lim_{n\rightarrow\infty}\sup_{h\in
K^{\prime}}\mathrm{E}^{Y_{\varepsilon}}\mathrm{E}_{\theta_{n,h}}\left[
Z_{\theta_{n,h},-h}\,\left(  \ell_{\theta_{n,h}}^{n}\right)  :\left\Vert
L_{\theta_{n,h}}^{n,\varepsilon}\right\Vert >a\,,\,\left\Vert Y_{\varepsilon
}\right\Vert \leq\varepsilon^{1/4}\right] \nonumber\\
&  \leq\lim_{\varepsilon\downarrow0}\lim_{n\rightarrow\infty}\sup_{h\in
K^{\prime}}\mathrm{E}_{\theta_{n,h}}^{n}\left[  I_{\left\{  \left\Vert
J_{\theta_{n,h}}^{-1}\ell_{\theta_{n,h}}^{n}\right\Vert >a-\varepsilon
^{1/4}\right\}  }Z_{\theta_{n,h},-h}\,\left(  \ell_{\theta_{n,h}}^{n}\right)
\right] \nonumber\\
&  \leq\lim_{\varepsilon\downarrow0}\lim_{n\rightarrow\infty}\sqrt{\sup_{h\in
K^{\prime}}P_{\theta_{n,h}}^{n}\left\{  \left\Vert \ell_{\theta_{n,h}}%
^{n}\right\Vert >\alpha_{\theta_{n,h}}\left(  a-\varepsilon^{1/4}\right)
\right\}  \mathrm{E}_{\theta_{n,h}}^{n}\left[  Z_{\theta_{n,h},-h}\,\left(
\ell_{\theta_{n,h}}^{n}\right)  \right]  ^{2}}\nonumber\\
&  \leq\lim_{\varepsilon\downarrow0}\lim_{n\rightarrow\infty}\sqrt{\sup_{h\in
K^{\prime}}P_{\theta_{n,h}}^{n}\left\{  \left\Vert \ell_{\theta_{n,h}}^{n_{2}%
}\right\Vert >\alpha_{\theta_{n,h}}\left(  a-\varepsilon^{1/4}\right)
\right\}  e^{2h^{T}J_{\theta_{n,h}}h}e^{-h^{T}J_{\theta_{n,h}}h}}\nonumber\\
&  \leq\lim_{\varepsilon\downarrow0}\lim_{n\rightarrow\infty}\sqrt{\sup_{h\in
K^{\prime}}\frac{1}{\left\{  \alpha_{\theta_{n,h}}\left(  a-\varepsilon
^{1/4}\right)  \right\}  ^{2}}\mathrm{E}_{\theta_{n,h}}^{n}\left\Vert
\ell_{\theta_{n,h}}^{n}\right\Vert ^{2}e^{h^{T}J_{\theta_{n,h}}h}}\nonumber\\
&  \leq\sqrt{\frac{\mathrm{tr}\,J_{\theta}}{\alpha_{\theta}^{2}a^{2}}%
\sup_{h\in K^{\prime}}e^{h^{T}J_{\theta}h}}, \label{EfZ-3}%
\end{align}
where $\alpha_{\theta}$ is the minimal eigenvalue of $J_{\theta}$.

After all, combining (\ref{EfZ}), (\ref{EfZ-EfX}), (\ref{EfZ-2}) and
(\ref{EfZ-3}), we have
\begin{align*}
&  \lim_{\varepsilon\downarrow0}\lim_{n\rightarrow\infty}\sup_{h\in K^{\prime
}}\left\vert \mathrm{E}^{Y_{\varepsilon}}\mathrm{E}_{\theta}^{n}\left[
f\left(  L_{\theta_{n,h}}^{n,\varepsilon}\right)  \right]  -\mathrm{E}\left[
f\left(  X_{-h}\right)  \right]  \right\vert \\
&  \leq\lim_{\varepsilon\downarrow0}\lim_{n\rightarrow\infty}\sup_{h\in
K^{\prime}}\left\vert \mathrm{E}^{Y_{\varepsilon}}\mathrm{E}_{\theta_{n,h}%
}^{n}\left[  f\left(  L_{\theta_{n,h}}^{n,\varepsilon}\right)  Z_{\theta
_{n,h},-h}\,\left(  \ell_{\theta_{n,h}}^{n}\right)  \right]  -\mathrm{E}%
\left[  f\left(  X_{-h}\right)  \right]  \right\vert \\
&  =\sqrt{\frac{\mathrm{tr}\,J_{\theta}}{\alpha_{\theta}^{2}a^{2}}\sup_{h\in
K^{\prime}}e^{h^{T}J_{\theta}h}}+\sup_{h\in K^{\prime}}\Pr\left\{  \left\Vert
X_{-h}\right\Vert >a\right\}  .
\end{align*}
Since $a$ is arbitrary, letting $a\rightarrow\infty$, we have
\[
\lim_{\varepsilon\downarrow0}\lim_{n\rightarrow\infty}\sup_{h\in K^{\prime}%
}\left\Vert \mathcal{L}\left(  L_{\theta_{n,h}}^{n,\varepsilon}|P_{\theta}%
^{n}\right)  -\mathrm{N}\left(  -h,J_{\theta}^{-1}\right)  \right\Vert
_{1}=0.
\]

\end{proof}

\begin{theorem}
\label{th:realize-by-gauss}Suppose $\theta\rightarrow p_{\theta}$ is
continuously differentiable in quadratic mean, and (\ref{EZ<infty}) holds.
Moreover, we suppose (\ref{J-cont}), (\ref{J-nondeg}), and (\ref{consistent})
hold. Then, we have (\ref{achieve}).
\end{theorem}

\begin{proof}
\ \ \ Observe, for any compact set $K^{\prime}\subset%
\mathbb{R}
^{m}$, \thinspace\ \
\begin{align}
&  \left\Vert \Lambda_{\delta,\varepsilon}^{n,r}\left(  P_{\theta}^{n}\right)
-P_{\theta}^{rn}\right\Vert _{1}.\nonumber\\
&  \leq\left\Vert \mathrm{E}^{\hat{\theta}^{n_{1}}}\mathrm{E}^{\tilde{X}%
_{\hat{\theta}^{n_{1}}}^{n}}R_{\hat{\theta}^{n_{1}}}^{rn}\left(  \cdot
|\tilde{X}_{\hat{\theta}^{n_{1}}}^{n}\right)  -\tilde{P}_{\theta}%
^{rn}\right\Vert _{1}\nonumber\\
&  \leq\mathrm{E}^{\hat{\theta}^{n_{1}}}\left\Vert \mathrm{E}^{\tilde{X}%
_{\hat{\theta}^{n_{1}}}^{n}}R_{\hat{\theta}^{n_{1}}}^{rn}\left(  \cdot
|\tilde{X}_{\hat{\theta}^{n_{1}}}^{n}\right)  -\tilde{P}_{\theta}%
^{rn}\right\Vert _{1}\nonumber\\
&  \leq\mathrm{E}^{\hat{\theta}^{n_{1}}}\left[  \left\Vert \mathrm{E}%
^{\tilde{X}_{\hat{\theta}^{n_{1}}}^{n}}R_{\hat{\theta}^{n_{1}}}^{rn}\left(
\cdot|\tilde{X}_{\hat{\theta}^{n_{1}}}^{n}\right)  -\tilde{P}_{\theta}%
^{rn}\right\Vert _{1}:\sqrt{n_{2}}\left(  \hat{\theta}^{n_{1}}-\theta\right)
\in K^{\prime}\right] \nonumber\\
&  +\mathrm{E}^{\hat{\theta}^{n_{1}}}\left[  \left\Vert \mathrm{E}^{\tilde
{X}_{\hat{\theta}^{n_{1}}}^{n}}R_{\hat{\theta}^{n_{1}}}^{rn}\left(
\cdot|\tilde{X}_{\hat{\theta}^{n_{1}}}^{n}\right)  -\tilde{P}_{\theta}%
^{rn}\right\Vert _{1}:\sqrt{n_{2}}\left(  \hat{\theta}^{n_{1}}-\theta\right)
\not \in K^{\prime}\right]  \label{|LP-P|}%
\end{align}

The second term of (\ref{|LP-P|}) is evaluated as
\begin{align*}
&  \varlimsup_{n\rightarrow\infty}\mathrm{E}^{\hat{\theta}^{n_{1}}}\left\{
\left\Vert \mathrm{E}^{\tilde{X}_{\hat{\theta}^{n_{1}}}^{n}}R_{\hat{\theta
}^{n_{1}}}^{rn}\left(  \cdot|\tilde{X}_{\hat{\theta}^{n_{1}}}^{n}\right)
-\tilde{P}_{\theta}^{rn}\right\Vert _{1}:\sqrt{rn}\left(  \hat{\theta}^{n_{1}%
}-\theta\right)  \notin K^{\prime}\right\} \\
&  \leq2\varlimsup_{n\rightarrow\infty}P_{\theta}^{n_{1}}\left\{  \sqrt
{rn}\left(  \hat{\theta}^{n_{1}}-\theta\right)  \notin K^{\prime}\right\}  ,
\end{align*}
whose left hand side becomes arbitrarily small as $\,\,\,K^{\prime}\uparrow%
\mathbb{R}
^{m}$, due to (\ref{consistent}).

Next, we evaluate the first term of (\ref{|LP-P|}).
\begin{align}
&  \mathrm{E}^{\hat{\theta}^{n_{1}}}\left[  \left\Vert \mathrm{E}^{\tilde
{X}_{\hat{\theta}^{n_{1}}}^{n}}R_{\hat{\theta}^{n_{1}}}^{rn}\left(
\cdot|\tilde{X}_{\hat{\theta}^{n_{1}}}^{n}\right)  -\tilde{P}_{\theta}%
^{rn}\right\Vert _{1}:\sqrt{n_{2}}\left(  \hat{\theta}^{n_{1}}-\theta\right)
\in K^{\prime}\right]  \nonumber\\
&  \leq\mathrm{E}^{\hat{\theta}^{n_{1}}}\left[  \left\Vert \mathrm{E}%
^{\tilde{X}_{\hat{\theta}^{n_{1}}}^{n}}R_{\hat{\theta}^{n_{1}}}^{rn}\left(
\cdot|\tilde{X}_{\hat{\theta}^{n_{1}}}^{n}\right)  -Q_{\hat{\theta}^{n_{1}%
},\sqrt{rn}\left(  \theta-\hat{\theta}^{n_{1}}\right)  }^{rn}\right\Vert
_{1}:\sqrt{n_{2}}\left(  \hat{\theta}^{n_{1}}-\theta\right)  \in K^{\prime
}\right]  \nonumber\\
&  +\mathrm{E}^{\hat{\theta}^{n_{1}}}\left[  \left\Vert Q_{\hat{\theta}%
^{n_{1}},\sqrt{rn}\left(  \theta-\hat{\theta}^{n_{1}}\right)  }^{rn}-\tilde
{P}_{\theta}^{rn}\right\Vert _{1}:\sqrt{n_{2}}\left(  \hat{\theta}^{n_{1}%
}-\theta\right)  \in K^{\prime}\right]  .\label{||P-Q||}%
\end{align}
The first term of (\ref{||P-Q||}) is evaluated as follows. Let
\[
h:=\sqrt{n_{2}}\left(  \hat{\theta}^{n_{1}}-\theta\right)  =\sqrt
{\frac{1-\delta}{r}}\sqrt{rn}\left(  \hat{\theta}^{n_{1}}-\theta\right)  ,
\]
or equivalently,
\[
\hat{\theta}^{n_{1}}=\theta_{n_{2},h}=\theta_{rn,\tilde{h}},
\]
where $\tilde{h}:=\sqrt{r\left(  1-\delta\right)  ^{-1}}h$. Then,
\begin{align*}
&  \sup_{\tilde{A}}\left\vert \mathrm{E}^{\tilde{X}_{\hat{\theta}^{n_{1}}}%
^{n}}R_{\hat{\theta}^{n_{1}}}^{rn}\left(  \tilde{A}|\tilde{X}_{\hat{\theta
}^{n_{1}}}^{n}\right)  -Q_{\hat{\theta}^{n_{1}},\sqrt{rn}\left(  \theta
-\hat{\theta}^{n_{1}}\right)  }^{rn}\left(  \tilde{A}\right)  \right\vert \\
&  =\sup_{\tilde{A}}\left\vert \mathrm{E}^{\tilde{X}_{\hat{\theta}^{n_{1}}%
}^{n}}R_{\hat{\theta}^{n_{1}}}^{rn}\left(  \tilde{A}|\tilde{X}_{\hat{\theta
}^{n_{1}}}^{n}\right)  -Q_{\hat{\theta}^{n_{1}},-\tilde{h}}^{rn}\left(
\tilde{A}\right)  \right\vert \\
&  =\sup_{\tilde{A}}\left\vert \mathrm{E}^{\tilde{X}_{\hat{\theta}^{n_{1}}%
}^{n}}R_{\hat{\theta}^{n_{1}}}^{rn}\left(  \tilde{A}|\tilde{X}_{\hat{\theta
}^{n_{1}}}^{n}\right)  -\mathrm{E}^{\lambda_{-\tilde{h}}^{rn}}R_{\hat{\theta
}^{n_{1}}}^{rn}\left(  \tilde{A}|\lambda_{-\tilde{h}}^{rn}\right)  \right\vert
\\
&  \leq\left\Vert \mathcal{L}\left(  \tilde{X}_{\hat{\theta}^{n_{1}}}%
^{n}\right)  -\mathrm{N}\left(  -\tilde{h},J_{\theta}^{-1}\right)  \right\Vert
_{1}\\
&  \leq\left\Vert \mathcal{L}\left(  \tilde{X}_{\hat{\theta}^{n_{1}}}%
^{n}\right)  -\Lambda_{\mathrm{amp}}^{\sqrt{r/\left(  1-\delta\right)  }%
}\left(  \mathrm{N}\left(  -h,J_{\theta}^{-1}\right)  \right)  \right\Vert
_{1}+\left\Vert \Lambda_{\mathrm{amp}}^{\sqrt{r/\left(  1-\delta\right)  }%
}\left(  \mathrm{N}\left(  -h,J_{\theta}^{-1}\right)  \right)  -\mathrm{N}%
\left(  -\tilde{h},J_{\theta}^{-1}\right)  \right\Vert _{1}\\
&  =\left\Vert \Lambda_{\mathrm{amp}}^{\sqrt{r/\left(  1-\delta\right)  }%
}\left(  \mathcal{L}\left(  L_{\hat{\theta}^{n_{1}}}^{n_{2},\varepsilon
}|P_{\theta}^{n_{2}}\right)  \right)  -\Lambda_{\mathrm{amp}}^{\sqrt{r/\left(
1-\delta\right)  }}\left(  \mathrm{N}\left(  -h,J_{\theta}^{-1}\right)
\right)  \right\Vert _{1}\\
&  +\left\Vert \Lambda_{\mathrm{amp}}^{\sqrt{r/\left(  1-\delta\right)  }%
}\left(  \mathrm{N}\left(  -h,J_{\theta}^{-1}\right)  \right)  -\mathrm{N}%
\left(  -\sqrt{r/\left(  1-\delta\right)  }h,J_{\theta}^{-1}\right)
\right\Vert _{1}\\
&  \leq\left\Vert \mathcal{L}\left(  L_{\hat{\theta}^{n_{1}}}^{n_{2}%
,\varepsilon}|P_{\theta}^{n_{2}}\right)  -\mathrm{N}\left(  -h,J_{\theta}%
^{-1}\right)  \right\Vert _{1}+\sup_{h^{\prime}\in%
\mathbb{R}
^{m}}\left\Vert \Lambda_{\mathrm{amp}}^{\sqrt{r/\left(  1-\delta\right)  }%
}\left(  \mathrm{N}\left(  h^{\prime},J_{\theta}^{-1}\right)  \right)
-\mathrm{N}\left(  \sqrt{r/\left(  1-\delta\right)  }h^{\prime},J_{\theta
}^{-1}\right)  \right\Vert _{1}.
\end{align*}
Therefore, due to Lemma\thinspace\ref{lem:||p-N||},
\begin{align}
&  \lim_{\varepsilon\downarrow0}\varlimsup_{n\rightarrow\infty}\mathrm{E}%
^{\hat{\theta}^{n_{1}}}\left[  \sup_{\tilde{A}}\left\vert \mathrm{E}%
^{\tilde{X}_{\hat{\theta}^{n_{1}}}^{n}}R_{\hat{\theta}^{n_{1}}}^{rn}\left(
\tilde{A}|\tilde{X}_{\hat{\theta}^{n_{1}}}^{n}\right)  -Q_{\hat{\theta}%
^{n_{1}},\sqrt{rn}\left(  \theta-\hat{\theta}^{n_{1}}\right)  }^{rn}\left(
\tilde{A}\right)  \right\vert ;\sqrt{n_{2}}\left(  \hat{\theta}^{n_{1}}%
-\theta\right)  \in K^{\prime}\right]  \nonumber\\
&  \leq\lim_{\varepsilon\downarrow0}\varlimsup_{n\rightarrow\infty}\sup_{h\in
K^{\prime}}\left\Vert \mathcal{L}\left(  L_{\theta_{n_{2},h}}^{n_{2}%
,\varepsilon}|P_{\theta}^{n_{2}}\right)  -\mathrm{N}\left(  -h,J_{\theta}%
^{-1}\right)  \right\Vert _{1}\nonumber\\
&  +\sup_{h\in K^{\prime}}\sup_{h^{\prime}\in%
\mathbb{R}
^{m}}\left\Vert \Lambda_{\mathrm{amp}}^{\sqrt{r/\left(  1-\delta\right)  }%
}\left(  \mathrm{N}\left(  h^{\prime},J_{\theta}^{-1}\right)  \right)
-\mathrm{N}\left(  \sqrt{r/\left(  1-\delta\right)  }h^{\prime},J_{\theta
}^{-1}\right)  \right\Vert _{1}\nonumber\\
&  =\sup_{h\in%
\mathbb{R}
^{m}}\left\Vert \Lambda_{\mathrm{amp}}^{\sqrt{r/\left(  1-\delta\right)  }%
}\left(  \mathrm{N}\left(  h,J_{\theta}^{-1}\right)  \right)  -\mathrm{N}%
\left(  \sqrt{r/\left(  1-\delta\right)  }h,J_{\theta}^{-1}\right)
\right\Vert _{1}\nonumber\\
&  =D_{r/\left(  1-\delta\right)  ,J_{\theta}^{-1}}.\label{||P-Q||-1}%
\end{align}

The second term of (\ref{||P-Q||}) is evaluated as follows. Let $K^{\prime}$
be an arbitrary compact set in $%
\mathbb{R}
^{m}$ and $K_{\theta}$ be an arbitrary compact set in $\Theta$ with $\theta\in
K$. Then, due to (\ref{discrete}) and Theorem\thinspace\ref{th:uniform-lan},%
\begin{align}
&  \mathrm{E}^{\hat{\theta}^{n_{1}}}\left\{  \left\Vert Q_{\hat{\theta}%
^{n_{1}},\sqrt{rn}\left(  \theta-\hat{\theta}^{n_{1}}\right)  }^{rn}-\tilde
{P}_{\theta}^{rn}\right\Vert _{1}:\sqrt{rn}\left(  \hat{\theta}^{n_{1}}%
-\theta\right)  \in K^{\prime}\right\} \nonumber\\
&  \leq\sup_{\theta^{\prime}\in K_{\theta}}\sup_{h\in K^{\prime}}\left\Vert
Q_{\theta^{\prime},\,-h}^{rn}-\tilde{P}_{\theta^{\prime}-h/\sqrt{rn}}%
^{rn}\right\Vert _{1}\,\rightarrow0,\,\,n\rightarrow\infty. \label{||P-Q||-2}%
\end{align}
Therefore, combining (\ref{||P-Q||-1}) and (\ref{||P-Q||-2}), we have%
\[
\lim_{\varepsilon\downarrow0}\varlimsup_{n\rightarrow\infty}\mathrm{E}%
^{\hat{\theta}^{n_{1}}}\left[  \left\Vert \mathrm{E}^{\tilde{X}_{\hat{\theta
}^{n_{1}}}^{n}}R_{\hat{\theta}^{n_{1}}}^{rn}\left(  \cdot|\tilde{X}%
_{\hat{\theta}^{n_{1}}}^{n}\right)  -\tilde{P}_{\theta}^{rn}\right\Vert
_{1}:\sqrt{n_{2}}\left(  \hat{\theta}^{n_{1}}-\theta\right)  \in K^{\prime
}\right]  \leq D_{r/\left(  1-\delta\right)  ,J_{\theta}^{-1}}%
\]

After all, we have%
\[
\lim_{\varepsilon\downarrow0}\lim_{n\rightarrow\infty}\left\Vert
\Lambda_{\delta,\varepsilon}^{n,r}\left(  P_{\theta}^{n}\right)  -P_{\theta
}^{rn}\right\Vert \leq D_{r/\left(  1-\delta\right)  ,J_{\theta}^{-1}}.
\]
Since the map $x\rightarrow D_{x,\Sigma}$ is continuous about $x$ (see
(\ref{gauss-opt-loss}) ), letting $\delta\rightarrow0$, we have the asserted result.
\end{proof}

\subsection{Local minimax property}

Based on $\left(  n,rn\right)  $-cloner $\Lambda^{n,r}$ for $\left\{
P_{\theta}^{n}\right\}  _{\theta\in\Theta}$, we compose an amplifier
$\Lambda_{\mathrm{amp}}^{\theta,\sqrt{r},n,\varepsilon}$ for the Gaussian
shift $\left\{  \mathrm{N}\left(  h,J_{\theta}^{-1}\right)  \right\}  _{h\in%
\mathbb{R}
^{m}}$ as follows.

\begin{description}
\item[(I)] Given $X_{h}$ with $\mathcal{L}\left(  X_{h}\right)  =\mathrm{N}%
\left(  h,J_{\theta}^{-1}\right)  $, compose
\[
Q_{\theta,h}^{\prime n}\left(  A\right)  :=Q_{\theta,h}^{n}\left(
A\times\Omega^{\prime}\right)  =\mathrm{E}^{X_{h}}R_{\theta}^{n}\left(
A\times\Omega^{\prime}|X_{h}\right)  .
\]

\item[(II)] Apply $\Lambda^{n,r}$ to $Q_{\theta,h}^{\prime n}$. \ Denote by
$W_{h}^{n,r}$ the random variable with $\mathcal{L}\left(  W_{h}^{n,r}\right)
=$ $\Lambda^{n,r}\left(  Q_{\theta,h}^{\prime n}\right)  $.

\item[(III)] The output random variable is $\tilde{X}_{h}^{r,n,\varepsilon
}:=L_{\theta}^{rn,\varepsilon}\left(  W_{h}^{n,r}\right)  $, where
$\ L_{\theta}^{n,\varepsilon}\left(  \omega^{n}\right)  :=J_{\theta}^{-1}%
\ell_{\theta}^{n}\left(  \omega^{n}\right)  +Y_{\varepsilon}$ and
$\mathcal{L}\left(  Y_{\varepsilon}\right)  =\mathrm{N}\left(  0,\varepsilon
\right)  $.
\end{description}

\begin{lemma}
\label{lem:||p-N||-1}Suppose $\theta\rightarrow p_{\theta}$ is differentiable
in quadratic mean, and (\ref{EZ<infty}) and (\ref{J-nondeg}) hold. Then, for
any compact set $K^{\prime}\in%
\mathbb{R}
^{m}$, we have%
\[
\lim_{n\rightarrow\infty}\sup_{h\in K^{\prime}}\left\Vert \mathcal{L}\left(
L_{\theta}^{n,\varepsilon}|P_{\theta_{n,h}}^{n}\right)  -\mathrm{N}\left(
h,J_{\theta}^{-1}\right)  \right\Vert _{1}=0.
\]

\end{lemma}

\begin{proof}
Observe, for any measurable function $f$ with $\sup_{x\in%
\mathbb{R}
^{m}}\left\vert f\left(  x\right)  \right\vert \leq1$,%
\[
\mathrm{E}^{Y_{\varepsilon}}\mathrm{E}_{\theta_{n,h}}^{n}\left[  f\left(
L_{\theta}^{n,\varepsilon}\right)  \right]  =\mathrm{E}^{Y_{\varepsilon}%
}\mathrm{E}_{\theta}^{n}\left[  f\left(  L_{\theta}^{n,\varepsilon}\right)
Z_{\theta,h}^{n}\,\right]  ,
\]
Observe also, due to Lemma\thinspace\ref{lem:EZ-EZl-1}, with $K^{\prime}$
being an arbitrary compact subset of $%
\mathbb{R}
^{m}$,%
\begin{align}
&  \sup_{h\in K^{\prime}}\left\vert \mathrm{E}^{Y_{\varepsilon}}%
\mathrm{E}_{\theta}^{n}\left[  f\left(  L_{\theta}^{n,\varepsilon}\right)
\left(  Z_{\theta,h}^{n}\,-Z_{\theta,h}\left(  \ell_{\theta}^{n}\right)
\,\right)  \right]  \right\vert \label{EfZZ}\\
&  \leq\sup_{h\in K^{\prime}}\mathrm{E}_{\theta_{n,h}}^{n}\left[  \left\vert
Z_{\theta,h}^{n}\,-Z_{\theta,h}\left(  \ell_{\theta}^{n}\right)  \right\vert
\,\right] \nonumber\\
&  \rightarrow0,\,\,n\rightarrow\infty.\nonumber
\end{align}
Therefore, we have to evaluate
\begin{align*}
&  \lim_{\varepsilon\downarrow0}\lim_{n\rightarrow\infty}\sup_{h\in K^{\prime
}}\mathrm{E}^{Y_{\varepsilon}}\mathrm{E}_{\theta}^{n}\left[  f\left(
L_{\theta}^{n,\varepsilon}\right)  Z_{\theta,h}\,\left(  \ell_{\theta}%
^{n}\right)  \right] \\
&  =\lim_{\varepsilon\downarrow0}\lim_{n\rightarrow\infty}\sup_{h\in
K^{\prime}}\mathrm{E}^{Y_{\varepsilon}}\mathrm{E}_{\theta}^{n}\left[  f\left(
L_{\theta}^{n,\varepsilon}\right)  Z_{\theta,h}\,\left(  J_{\theta}L_{\theta
}^{n,\varepsilon}\right)  e^{-h^{T}J_{\theta}Y_{\varepsilon}}\right] \\
&  =\lim_{\varepsilon\downarrow0}\lim_{n\rightarrow\infty}\sup_{h\in
K^{\prime}}\left(  E_{1}+E_{2}+E_{3}\right)  ,
\end{align*}
where
\begin{align*}
E_{1}  &  :=\mathrm{E}^{Y_{\varepsilon}}\mathrm{E}_{\theta}^{n}\left[
I_{\left\{  \left\Vert L_{\theta}^{n,\varepsilon}\right\Vert \leq
a\,,\,\left\Vert Y_{\varepsilon}\right\Vert \leq\varepsilon^{1/4}\right\}
}f\left(  L_{\theta}^{n,\varepsilon}\right)  Z_{\theta,h}\,\left(  J_{\theta
}L_{\theta}^{n,\varepsilon}\right)  e^{-h^{T}J_{\theta}Y_{\varepsilon}%
}\right]  ,\\
E_{2}  &  :=\mathrm{E}^{Y_{\varepsilon}}\mathrm{E}_{\theta}^{n}\left[
I_{\left\{  \left\Vert Y_{\varepsilon}\right\Vert >\varepsilon^{1/4}\right\}
}f\left(  L_{\theta}^{n,\varepsilon}\right)  Z_{\theta,h}\,\left(  J_{\theta
}L_{\theta}^{n,\varepsilon}\right)  e^{-h^{T}J_{\theta}Y_{\varepsilon}%
}\right]  ,\\
E_{3}  &  :=\mathrm{E}^{Y_{\varepsilon}}\mathrm{E}_{\theta}^{n}\left[
I_{\left\{  \left\Vert L_{\theta}^{n,\varepsilon}\right\Vert >a\,,\,\left\Vert
Y_{\varepsilon}\right\Vert \leq\varepsilon^{1/4}\right\}  }f\left(  L_{\theta
}^{n,\varepsilon}\right)  Z_{\theta,h}\,\left(  J_{\theta}L_{\theta
}^{n,\varepsilon}\right)  e^{-h^{T}J_{\theta}Y_{\varepsilon}}\right]  .
\end{align*}

$E_{1}$ is evaluated as follows. Observe%
\[
E_{1}=\mathrm{E}^{Y_{\varepsilon}}\mathrm{E}_{\theta_{n,h}}^{n}\left[
I_{\left\{  \left\Vert L_{\theta}^{n,\varepsilon}\right\Vert \leq a\right\}
}f\left(  L_{\theta}^{n,\varepsilon}\right)  Z_{\theta,h}\,\left(  J_{\theta
}L_{\theta}^{n,\varepsilon}\right)  \mathrm{E}\left[  \left.  I_{\left\{
\left\Vert Y_{\varepsilon}\right\Vert \leq\varepsilon^{1/4}\right\}
}e^{-h^{T}J_{\theta}Y_{\varepsilon}}\right\vert L_{\theta}^{n,\varepsilon
}\right]  \right]  ,
\]
whose second factor can be evaluated as
\begin{equation}
\lim_{\varepsilon\downarrow0}\lim_{n\rightarrow\infty}\sup_{h\in K^{\prime}%
}\left\vert \mathrm{E}\left[  \left.  I_{\left\{  \left\Vert Y_{\varepsilon
}\right\Vert \leq\varepsilon^{1/4}\right\}  }e^{-h^{T}J_{\theta}Y}\right\vert
L_{\theta_{n,h}}^{n,\varepsilon}\right]  -1\right\vert \leq e^{\left\Vert
h\right\Vert \left\Vert J_{\theta}\right\Vert \varepsilon^{1/4}}
\label{EfZZ-1-1}%
\end{equation}

To evaluate the first factor
\[
E_{1,1}:=\mathrm{E}^{Y_{\varepsilon}}\mathrm{E}_{\theta}^{n}\left[  f\left(
L_{\theta}^{n,\varepsilon}\right)  Z_{\theta,h}\,\left(  J_{\theta}L_{\theta
}^{n,\varepsilon}\right)  :\left\Vert L_{\theta}^{n,\varepsilon}\right\Vert
\leq a\right]  ,
\]
observe, with $\mathcal{L}\left(  X_{h}\right)  =\mathrm{N}\left(
h,J_{\theta}^{-1}\right)  $,%
\begin{align}
&  \lim_{\varepsilon\downarrow0}\lim_{n\rightarrow\infty}\sup_{h\in K^{\prime
}}\left\vert E_{1,1}-\mathrm{E}\left[  f\left(  X_{-h}\right)  :\left\Vert
X_{-h}\right\Vert \leq a\right]  \right\vert \nonumber\\
&  =\lim_{\varepsilon\downarrow0}\lim_{n\rightarrow\infty}\sup_{h\in
K^{\prime}}\left\vert
\begin{array}
[c]{c}%
\mathrm{E}^{Y_{\varepsilon}}\mathrm{E}_{\theta}^{n}\left[  f\left(  L_{\theta
}^{n,\varepsilon}\right)  Z_{\theta,h}\,\left(  J_{\theta}L_{\theta
}^{n,\varepsilon}\right)  :\left\Vert L_{\theta}^{n,\varepsilon}\right\Vert
\leq a\right] \\
-\mathrm{E}^{X_{0}}\left[  f\left(  X_{0}\right)  Z_{\theta,h}\,\left(
J_{\theta}X_{0}\right)  :\left\Vert X_{0}\right\Vert \leq a\right]
\end{array}
\right\vert \nonumber\\
&  \leq\sup_{h\in K^{\prime}}e^{\left\Vert h\right\Vert \left\Vert J_{\theta
}\right\Vert a-\frac{1}{2}h^{T}J_{\theta}h}\lim_{\varepsilon\downarrow0}%
\lim_{n\rightarrow\infty}\left\Vert \mathcal{L}\left(  L_{\theta
}^{n,\varepsilon}|P_{\theta}^{n}\right)  -\mathrm{N}\left(  0,J_{\theta}%
^{-1}\right)  \right\Vert _{1}\nonumber\\
&  =0, \label{EfZZ-1-2}%
\end{align}
where the last identity is due to Lemma\thinspace\ref{lem:local-limit}.

Therefore, by (\ref{EfZZ-1-1}) and (\ref{EfZZ-1-2}),
\begin{equation}
\lim_{\varepsilon\downarrow0}\lim_{n\rightarrow\infty}\sup_{h\in K^{\prime}%
}\left\vert E_{1}-\mathrm{E}\left[  f\left(  X_{-h}\right)  :\left\Vert
X_{-h}\right\Vert \leq a\right]  \right\vert =0. \label{EfZZ-1}%
\end{equation}

On the other hand, by (\ref{EZ<finite}),%
\begin{align}
&  \lim_{\varepsilon\downarrow0}\lim_{n\rightarrow\infty}\sup_{h\in K^{\prime
}}E_{2}\nonumber\\
&  \leq\lim_{\varepsilon\downarrow0}\lim_{n\rightarrow\infty}\sup_{h\in
K^{\prime}}\mathrm{E}^{Y_{\varepsilon}}\mathrm{E}_{\theta}^{n}\left[
Z_{\theta,h}\,\left(  \ell_{\theta}^{n}\right)  :\left\Vert Y_{\varepsilon
}\right\Vert >\varepsilon^{1/4}\right] \nonumber\\
&  =\lim_{\varepsilon\downarrow0}\Pr\left\{  \left\Vert Y_{\varepsilon
}\right\Vert >\varepsilon^{1/4}\right\}  \cdot\sup_{h\in K^{\prime}}%
e^{h^{T}J_{\theta}h}e^{-\frac{1}{2}h^{T}J_{\theta}h}\nonumber\\
&  =0, \label{EfZZ-2}%
\end{align}
and
\begin{align}
&  \lim_{\varepsilon\downarrow0}\lim_{n\rightarrow\infty}\sup_{h\in K^{\prime
}}E_{3}\nonumber\\
&  \leq\lim_{\varepsilon\downarrow0}\lim_{n\rightarrow\infty}\sup_{h\in
K^{\prime}}\mathrm{E}^{Y_{\varepsilon}}\mathrm{E}_{\theta}\left[  Z_{\theta
,h}\,\left(  \ell_{\theta}^{n}\right)  :\left\Vert L_{\theta}^{n,\varepsilon
}\right\Vert >a\,,\,\left\Vert Y_{\varepsilon}\right\Vert \leq\varepsilon
^{1/2}\right] \nonumber\\
&  \leq\lim_{\varepsilon\downarrow0}\lim_{n\rightarrow\infty}\sup_{h\in
K^{\prime}}\mathrm{E}_{\theta}^{n}\left[  I_{\left\{  \left\Vert \ell_{\theta
}^{n}\right\Vert >\alpha_{\theta}\left(  a-\varepsilon^{1/4}\right)  \right\}
}Z_{\theta,h}\,\left(  \ell_{\theta}^{n}\right)  \right] \nonumber\\
&  \leq\lim_{\varepsilon\downarrow0}\lim_{n\rightarrow\infty}\sup_{h\in
K^{\prime}}\sqrt{P_{\theta}^{n}\left\{  \left\Vert \ell_{\theta}%
^{n}\right\Vert >\alpha_{\theta}\left(  a-\varepsilon^{1/4}\right)  \right\}
}\sqrt{\mathrm{E}_{\theta}^{n}\left[  Z_{\theta,h}\,\left(  \ell_{\theta}%
^{n}\right)  \right]  ^{2}}\nonumber\\
&  \leq\lim_{\varepsilon\downarrow0}\lim_{n\rightarrow\infty}\sqrt{\sup_{h\in
K^{\prime}}P_{\theta}^{n}\left\{  \left\Vert \ell_{\theta}^{n}\right\Vert
>\alpha_{\theta}\left(  a-\varepsilon^{1/4}\right)  \right\}  \sup_{h\in
K^{\prime}}\left(  \mathrm{E}_{\theta^{\prime}}^{n}e^{-\frac{2h^{T}}{\sqrt{n}%
}\ell_{\theta}}\right)  ^{n}e^{-h^{T}J_{\theta}h}}\nonumber\\
&  \leq\lim_{\varepsilon\downarrow0}\lim_{n\rightarrow\infty}\sqrt{\sup_{h\in
K^{\prime}}\frac{1}{\left\{  \alpha_{\theta}\left(  a-\varepsilon
^{1/4}\right)  \right\}  ^{2}}\mathrm{E}_{\theta}^{n}\left\Vert \ell_{\theta
}^{n}\right\Vert ^{2}\sup_{h\in K^{\prime}}e^{h^{T}J_{\theta}h}}\nonumber\\
&  =\sqrt{\sup_{h\in K^{\prime}}\frac{\mathrm{tr}\,J_{\theta}}{\alpha_{\theta
}^{2}a^{2}}\sup_{h\in K^{\prime}}e^{h^{T}J_{\theta}h}}. \label{EfZZ-3}%
\end{align}

After all, by (\ref{EfZZ}), (\ref{EfZZ-1}), (\ref{EfZ-2}), and (\ref{EfZ-3}),
we have%
\begin{align*}
&  \lim_{\varepsilon\downarrow0}\lim_{n\rightarrow\infty}\sup_{h\in K^{\prime
}}\left\vert \mathrm{E}^{Y_{\varepsilon}}\mathrm{E}_{\theta_{h_{n}}}%
^{n}\left[  f\left(  L_{\theta}^{n,\varepsilon}\right)  \right]
-\mathrm{E}\left[  f\left(  X_{h}\right)  \right]  \right\vert \\
&  =\lim_{\varepsilon\downarrow0}\lim_{n\rightarrow\infty}\sup_{h\in
K^{\prime}}\left\vert \mathrm{E}^{Y_{\varepsilon}}\mathrm{E}_{\theta}%
^{n}\left[  f\left(  L_{\theta}^{n,\varepsilon}\right)  Z_{\theta,h}\,\left(
\ell_{\theta}^{n}\right)  \right]  -\mathrm{E}\left[  f\left(  X_{h}\right)
\right]  \right\vert \\
&  =\sqrt{\sup_{h\in K^{\prime}}\frac{\mathrm{tr}\,J_{\theta}}{\alpha_{\theta
}^{2}a^{2}}\sup_{h\in K^{\prime}}e^{h^{T}J_{\theta}h}}+\sup_{h\in K^{\prime}%
}\Pr\left\{  \left\Vert X_{h}\right\Vert >a\right\}  .
\end{align*}
Since $a$ is arbitrary, letting $a\rightarrow\infty$, we have%
\[
\lim_{\varepsilon\downarrow0}\lim_{n\rightarrow\infty}\sup_{h\in K^{\prime}%
}\left\Vert \mathcal{L}\left(  L_{\theta}^{n,\varepsilon}|P_{\theta_{n,h}}%
^{n}\right)  -\mathrm{N}\left(  h,J_{\theta}^{-1}\right)  \right\Vert _{1}=0.
\]

\end{proof}

\begin{theorem}
Suppose  $\theta\rightarrow p_{\theta}$ is differentiable in quadratic mean,
and (\ref{EZ<infty}) and (\ref{J-nondeg}) hold. Then, for any Markov map
$\Lambda^{n,r}$ and for any $\theta\in\Theta$, we have (\ref{l-min-max}).
\end{theorem}

\begin{proof}
Observe
\begin{align*}
&  \left\Vert \Lambda_{\mathrm{amp}}^{\theta,\sqrt{r},n,\varepsilon}\left(
\mathrm{N}\left(  h,J_{\theta}^{-1}\right)  \right)  -\mathrm{N}\left(
\sqrt{r}h,J_{\theta}^{-1}\right)  \right\Vert _{1}\\
&  \leq\left\Vert \Lambda_{\mathrm{amp}}^{\theta,\sqrt{r},n,\varepsilon
}\left(  \mathrm{N}\left(  h,J_{\theta}^{-1}\right)  \right)  -\mathcal{L}%
\left(  L_{\theta}^{rn,\varepsilon}|P_{\theta_{n,h}}^{rn}\right)  \right\Vert
_{1}+\left\Vert \mathcal{L}\left(  L_{\theta}^{rn,\varepsilon}|P_{\theta
_{n,h}}^{rn}\right)  -\mathrm{N}\left(  \sqrt{r}h,J_{\theta}^{-1}\right)
\right\Vert _{1}\\
&  =\left\Vert \mathcal{L}\left(  L_{\theta}^{rn,\varepsilon}|\Lambda
^{n,r}\left(  Q_{\theta,h}^{\prime n}\right)  \right)  -\mathcal{L}\left(
L_{\theta}^{rn,\varepsilon}|P_{\theta_{n,h}}^{rn}\right)  \right\Vert
_{1}+\left\Vert \mathcal{L}\left(  L_{\theta}^{rn,\varepsilon}|P_{\theta
_{n,h}}^{rn}\right)  -\mathrm{N}\left(  \sqrt{r}h,J_{\theta}^{-1}\right)
\right\Vert _{1}\\
&  \leq\left\Vert \Lambda^{n,r}\left(  Q_{\theta,h}^{\prime n}\right)
-P_{\theta_{n,h}}^{rn}\right\Vert _{1}+\left\Vert \mathcal{L}\left(
L_{\theta}^{rn,\varepsilon}|P_{\theta_{n,h}}^{rn}\right)  -\mathrm{N}\left(
\sqrt{r}h,J_{\theta}^{-1}\right)  \right\Vert _{1}\\
&  \leq\left\Vert \Lambda^{n,r}\left(  Q_{\theta,h}^{\prime n}\right)
-\Lambda^{n,r}\left(  P_{\theta_{n,h}}^{n}\right)  \right\Vert _{1}+\left\Vert
\Lambda^{n,r}\left(  P_{\theta_{n,h}}^{n}\right)  -P_{\theta_{n,h}}%
^{rn}\right\Vert _{1}\\
&  +\left\Vert \mathcal{L}\left(  L_{\theta}^{rn,\varepsilon}|P_{\theta_{n,h}%
}^{rn}\right)  -\mathrm{N}\left(  \sqrt{r}h,J_{\theta}^{-1}\right)
\right\Vert _{1}\\
&  \leq\left\Vert Q_{\theta,h}^{\prime n}-P_{\theta_{n,h}}^{n}\right\Vert
_{1}+\left\Vert \Lambda^{n,r}\left(  P_{\theta_{n,h}}^{n}\right)
-P_{\theta_{n,h}}^{rn}\right\Vert _{1}+\left\Vert \mathcal{L}\left(
L_{\theta}^{rn,\varepsilon}|P_{\theta_{n,h}}^{rn}\right)  -\mathrm{N}\left(
\sqrt{r}h,J_{\theta}^{-1}\right)  \right\Vert _{1}.
\end{align*}
By Theorem\thinspace\ref{th:lan},the first term vanishes,
\[
\lim_{n\rightarrow\infty}\sup_{h\in K^{\prime}}\left\Vert Q_{\theta,h}^{\prime
n}-P_{\theta_{n,h}}^{n}\right\Vert _{1}=0.
\]
By Lemma \ref{lem:||p-N||-1}, since $\theta_{n,h}=\theta_{rn,\sqrt{r}h}$,
\ the third term vanishes,
\[
\lim_{\varepsilon\downarrow0}\lim_{n\rightarrow\infty}\sup_{h\in K^{\prime}%
}\left\Vert \mathcal{L}\left(  L_{\theta}^{rn,\varepsilon}|P_{\theta_{n,h}%
}^{rn}\right)  -\mathrm{N}\left(  \sqrt{r}h,J_{\theta}^{-1}\right)
\right\Vert _{1}=0.
\]
After all, we have
\[
\lim_{\varepsilon\downarrow0}\varliminf_{n\rightarrow\infty}\sup_{h\in
K^{\prime}}\left\Vert \Lambda_{\mathrm{amp}}^{\theta,\sqrt{r},n,\varepsilon
}\left(  \mathrm{N}\left(  h,J_{\theta}^{-1}\right)  \right)  -\mathrm{N}%
\left(  \sqrt{r}h,J_{\theta}^{-1}\right)  \right\Vert _{1}\leq\varliminf
_{n\rightarrow\infty}\sup_{h\in K^{\prime}}\left\Vert \Lambda^{n,r}\left(
P_{\theta_{n,h}}^{n}\right)  -P_{\theta_{n,h}}^{rn}\right\Vert _{1}.
\]
Since
\[
\inf_{\Lambda}\sup_{h\in K^{\prime}}\left\Vert \Lambda_{\mathrm{amp}}%
^{\theta,\sqrt{r}}\left(  \mathrm{N}\left(  h,J_{\theta}^{-1}\right)  \right)
-\mathrm{N}\left(  \sqrt{r}h,J_{\theta}^{-1}\right)  \right\Vert _{1}%
\leq\varliminf_{n\rightarrow\infty}\sup_{h\in K^{\prime}}\left\Vert
\Lambda_{\mathrm{amp}}^{\theta,\sqrt{r},n,\varepsilon}\left(  \mathrm{N}%
\left(  h,J_{\theta}^{-1}\right)  \right)  -\mathrm{N}\left(  \sqrt
{r}h,J_{\theta}^{-1}\right)  \right\Vert _{1}%
\]
holds due to optimality of $\Lambda_{\mathrm{amp}}^{\theta,\sqrt{r}}$, we
have
\[
\varliminf_{n\rightarrow\infty}\sup_{h\in K^{\prime}}\left\Vert \Lambda
^{n,r}\left(  P_{\theta_{n,h}}^{n}\right)  -P_{\theta_{n,h}}^{rn}\right\Vert
_{1}\geq\inf_{\Lambda}\sup_{h\in K^{\prime}}\left\Vert \Lambda_{\mathrm{amp}%
}^{\theta,\sqrt{r}}\left(  \mathrm{N}\left(  h,J_{\theta}^{-1}\right)
\right)  -\mathrm{N}\left(  \sqrt{r}h,J_{\theta}^{-1}\right)  \right\Vert
_{1}.
\]
Here, letting $K^{\prime}=\left\{  x\,;\left\Vert x\right\Vert \leq a\right\}
$ and $a\rightarrow\infty$, we have (\ref{l-min-max}).
\end{proof}

\section{Discussion}

Using quantum LAN, we can produce similar results for finite dimensional
quantum system.

\bigskip\appendix

\section{Proof of (\ref{inf|p-p|})}

In this appendix, we prove%
\begin{equation}
\inf_{x}\left\Vert p_{0,\mathbf{1}}-p_{x,r\mathbf{1}}\right\Vert
_{1}=\left\Vert p_{0,\mathbf{1}}\left(  y\right)  -p_{0,r\mathbf{1}}\left(
y\right)  \right\Vert _{1},\nonumber
\end{equation}
where $p_{x,\Sigma}$ is the density of $\mathrm{N}\left(  x,\Sigma\right)  $.
\ Due to the symmetry, we can suppose $x=\left(  t,0,\cdots,0\right)
$,\thinspace$t\geq0$. \ Define $B_{r,t}:=\left\{  y;p_{0,\mathbf{1}}\left(
y\right)  \geq p_{x,r\mathbf{1}}\left(  y\right)  \right\}  $. Observe
\begin{align*}
y  &  \in B_{r}\\
&  \Leftrightarrow\frac{1}{\left(  2\pi\right)  ^{m/2}}e^{-\frac{\left(
y_{1}\right)  ^{2}+\sum_{\kappa=2}^{m}\left(  y_{\kappa}\right)  ^{2}}{2}}%
\geq\frac{1}{\left(  2\pi r\right)  ^{m/2}}e^{-\frac{\left(  y_{1}-t\right)
^{2}+\sum_{\kappa=2}^{m}\left(  y_{\kappa}\right)  ^{2}}{2r}}\\
&  \Leftrightarrow\left(  y_{1}-\frac{t}{r-1}\right)  ^{2}+\sum_{\kappa=2}%
^{m}\left(  y_{\kappa}\right)  ^{2}\leq\frac{2m}{r-1}\log r+\left(  \frac
{t}{r-1}\right)  ^{2}.
\end{align*}
Hence, $B_{r}$ is a ball. \ 

For $z\in%
\mathbb{R}
^{m-1}$ and $t\in%
\mathbb{R}
$, define $t_{1}$ and $t_{2}$ by
\begin{align*}
p_{0,\mathbf{1}}\left(  \left(  t_{\kappa},z\right)  \right)   &
=p_{0,r\mathbf{1}}\left(  \left(  t_{\kappa}-t,z\right)  \right)  ,\,\,\\
t_{1}  &  \leq t_{2}.
\end{align*}
One can verify \
\[
p_{0,\mathbf{1}}\left(  \left(  t_{1},z\right)  \right)  \geq p_{0,\mathbf{1}%
}\left(  \left(  t_{2},z\right)  \right)
\]
as follows. In case of $0\leq t_{1}\left(  t,z\right)  \leq t_{2}\left(
t,z\right)  $, this holds because $p_{0,\mathbf{1}}\left(  \left(
\cdot,z\right)  \right)  $ is monotone decreasing on $%
\mathbb{R}
_{+}$. In case of $t_{1}\left(  t,z\right)  \leq0\leq t_{2}\left(  t,z\right)
$, observe $t_{1}\left(  0,z\right)  \leq t_{1}\left(  t,z\right)  \leq0\leq
t_{2}\left(  0,z\right)  \leq t_{2}\left(  t,z\right)  $. Hence,
\[
p_{0,\mathbf{1}}\left(  \left(  t_{1}\left(  t,z\right)  ,z\right)  \right)
\geq p_{0,\mathbf{1}}\left(  \left(  t_{1}\left(  0,z\right)  ,z\right)
\right)  =p_{0,\mathbf{1}}\left(  \left(  t_{2}\left(  0,z\right)  ,z\right)
\right)  \geq p_{0,\mathbf{1}}\left(  \left(  t_{2}\left(  t,z\right)
,z\right)  \right)  .
\]

Therefore,
\[
\frac{\mathrm{d}}{\mathrm{d}t}\left\Vert p_{0,\mathbf{1}}-p_{x,r\mathbf{1}%
}\right\Vert _{1}=\int\left\{  p_{0,\mathbf{1}}\left(  \left(  t_{1},z\right)
\right)  -p_{0,\mathbf{1}}\left(  \left(  t_{2},z\right)  \right)  \right\}
\mathrm{d}z\geq0.
\]
Therefore, the minimum is achieved $t=0$, and we have the asserted result.

\end{document}